\documentclass[lettersize,journal]{IEEEtran}
\usepackage{amsmath,amsfonts,amssymb}
\usepackage[lined,boxed,commentsnumbered]{algorithm2e}
\usepackage{array}
\usepackage[caption=false]{subfig}
\usepackage{textcomp}
\usepackage{subfloat}
\usepackage{url}
\usepackage{verbatim}
\usepackage{graphicx}
\usepackage{cite}
\usepackage[n,  % or lambda 
    advantage,
    operators,
    sets,
    adversary,
    landau,
    probability,
    notions,
    logic,
    ff, 
    mm,
    primitives,
    events,
    complexity,
    oracles,
    asymptotics,
    keys
]{cryptocode}
\renewcommand{\advantage}[2][n]{\textnormal{\textsf{Adv}}_{\mathrm{#1}}(#2)}
\renewcommand{\prob}[1]{\textnormal{\textsf{Pr}}[#1]}
\usepackage{xcolor}
\definecolor{myblue}{HTML}{0000CC}
\usepackage[colorlinks=true]{hyperref}
\hypersetup{
    linkcolor=myblue,
    urlcolor=myblue,
    citecolor=myblue
}
\usepackage{booktabs} 
\usepackage{multirow}
\usepackage{enumerate}
\usepackage{amsthm}
\usepackage{stmaryrd}
\usepackage{threeparttable}
\usepackage{bm}
\usepackage{ulem}

\createprocedureblock{procb}{center, boxed}{}{}{}
\newtheorem{theorem}{\bf Theorem}
\newtheorem{myDef}{Definition}

\def\BibTeX{{\rm B\kern-.05em{\sc i\kern-.025em b}\kern-.08em
    T\kern-.1667em\lower.7ex\hbox{E}\kern-.125emX}}
\usepackage{balance}

\begin{document}
\title{SOCI\textsuperscript{+}: An Enhanced Toolkit for Secure Outsourced Computation on Integers}
\author{Bowen Zhao, \textit{Member}, \textit{IEEE}, Weiquan Deng, Xiaoguo Li, Ximeng Liu, \textit{Senior Member, IEEE}, Qingqi Pei, \textit{Senior Member, IEEE}, Robert H. Deng, \textit{Fellow, IEEE}
\thanks{Manuscript received x x, x; revised x x, x.}}

% \markboth{Journal of \LaTeX\ Class Files,~Vol.~x, No.~x, x~xx}%
% {How to Use the IEEEtran \LaTeX \ Templates}

\maketitle

\begin{abstract}
Secure outsourced computation is critical for cloud computing to safeguard data confidentiality and ensure data usability.
Recently, secure outsourced computation schemes following a twin-server architecture based on partially homomorphic cryptosystems have received increasing attention.  
The Secure Outsourced Computation on Integers (SOCI) \cite{zhao2022soci} toolkit is the state-of-the-art among these schemes which can perform secure computation on integers without requiring the costly bootstrapping operation as in fully homomorphic encryption; however, SOCI suffers from relatively large computation and communication overhead.  
In this paper, we propose SOCI\textsuperscript{+} which significantly improves the performance of SOCI. 
Specifically, SOCI\textsuperscript{+} employs a novel $(2,2)$-threshold Paillier cryptosystem with fast encryption and decryption as its cryptographic primitive, and supports a suite of efficient secure arithmetic computation on integers protocols, including a secure multiplication protocol (SMUL), a secure comparison protocol (SCMP), a secure sign bit-acquisition protocol (SSBA), and a secure division protocol (SDIV), all based on the $(2,2)$-threshold Paillier cryptosystem with fast encryption and decryption. 
In addition, SOCI\textsuperscript{+} incorporates an offline and online computation mechanism to further optimize its performance. 
We perform rigorous theoretical analysis to prove the correctness and security of SOCI\textsuperscript{+}. 
Compared with SOCI, our experimental evaluation shows that SOCI\textsuperscript{+} is up to 5.4 times more efficient in computation and 40\% less in communication overhead.
\end{abstract}

\begin{IEEEkeywords}
Secure outsourced computation; Paillier cryptosystem; threshold cryptosystem; homomorphic encryption; secure computing.
\end{IEEEkeywords}

\section{Introduction}
\IEEEPARstart{C}{loud} computing provides flexible and convenient services for data outsourced computation, but it is prone to leak outsourced data.
Users with limited computation and storage capabilities can outsource their data to the cloud server and perform efficient computations over outsourced data 
\cite{shan2018practical}.
However, the cloud server may intentionally or unintentionally steal and leak the outsourced data, leading to privacy concerns.
At present, numerous data breaches have occured over the world.
For example, Facebook exposed a large amount of user data online for a fortnight due to misconfiguration in the cloud \cite{facebook}.
In addition, according to the data breach chronology published by \cite{Data_Breach_Chronology}, from 2005 to 2022, there have been about 20,000 instances of data breaches  in the United States, affecting approximately two billion records.

To prevent data leakages, users can encrypt data before outsourcing \cite{li2018privacy}.
However, performing computations over encrypted data (also known as ciphertext) is challenging, as conventional cryptosystems usually fail to enable computations over ciphertext directly. 
Secure outsourced computation is an effective manner balancing data security and data usability \cite{shan2018practical}, which enables computations on encrypted data directly.
Secure outsourced computation that ensures data security offers a promising computing paradigm for cloud computing, and it can be used in many fields, such as privacy-preserving machine learning training \cite{wang2020outsourced} and privacy-preserving evolutionary computation \cite{zhao2022evolution}. 

Homomorphic cryptosystems enable secure outsourced computation as their features achieving addition, multiplication, or both of addition and multiplication over ciphertext.
Unfortunately, secure outsourced computation based on homomorphic cryptosystems still suffers from several challenges. 
Secure outsourced computation solely based on homomorphic cryptosystems is challenging to achieve nonlinear operations (e.g., comparison) \cite{zhao2022soci} and obtain the intermediate result. In certain scenarios, it is necessary to obtain the intermediate result, such as privacy-preserving person re-identification \cite{zhao2023identifiable}.
Homomorphic cryptosystems, such as fully homomorphic encryption that supports addition and multiplication over ciphertext simultaneously, suffer from significant storage costs \cite{zhao2022soci}. 
Partially homomorphic encryption (PHE) supports addition or multiplication over ciphertext and has a less ciphertext size.
To mitigate the limitations imposed by restricted computation types and high storage costs, a combination of PHE and a twin-server architecture has emerged as a promising and increasingly popular paradigm.
Despite these advancements, secure outsourced computation solutions \cite{zhao2022soci,liu2016privacy} based on PHE and the twin-server architecture still bring slightly high computation costs and communication costs.

To tackle the above challenges, in this paper, we propose an enhanced toolkit for secure outsourced  computation on integers, named SOCI\textsuperscript{+}, which is inspired by SOCI (a toolkit for secure outsourced computation on integers) \cite{zhao2022soci}. Building on the Paillier cryptosystem with fast encryption and decryption \cite{ma2021optimized}, we propose a novel $(2,2)$-threshold Paillier cryptosystem to mitigate computation costs. Subsequently, we redesign all secure computation protocols proposed by SOCI \cite{zhao2022soci}. 
Additionally, considering the underlying features of secure outsourced computation protocols, which allow a multitude of pre-encryption processes, we divide the computations of these protocols into two phases: offline phase and online phase. In short, the contributions of this paper are three-fold.
\begin{itemize}
    \item \textbf{A novel (2, 2)-threshold Paillier cryptosystem (FastPaiTD).}
    For the first time, we propose a novel $(2,2)$-threshold Paillier cryptosystem called FastPaiTD, which is based on the Paillier cryptosystem with fast encryption and decryption \cite{ma2021optimized}.
    FastPaiTD is specially designed to seamlessly adapt to a twin-server architecture. 
    \item \textbf{An offline and online mechanism.}
    To expedite the computations of secure computation protocols, we introduce an offline and online mechanism. Specifically, the encryption of random numbers and some constants in secure computation protocols are computed in advance at the offline phase, while the online phase only perform operations except for these operations performed offline.
    \item \textbf{A suite of secure computation protocols with superior performance.}
    To support linear operations and nonlinear operations,
    inspired by SOCI \cite{zhao2022soci},
    we adopt the proposed FastPaiTD to design a suite secure computation protocols, including 
    a secure multiplication protocol (\texttt{SMUL}), a secure comparison protocol (\texttt{SCMP}), a secure sign bit-acquisition protocol (\texttt{SSBA}), and a secure division protocol (\texttt{SDIV}).
    Compared with SOCI \cite{zhao2022soci}, our proposed protocols can improve up to 5.4 times in computation efficiency and saves up to $40\%$ in communication overhead.
\end{itemize}

The rest of this paper is organized as follows.
We briefly review related work in Section \ref{Section_2}, and show the preliminaries for constructing SOCI\textsuperscript{+} in Section \ref{Section_3}.
The system model and threat model are given in Section \ref{Section_4}.
In Section \ref{Section_5}, we firstly present the proposed $(2, 2)$-threshold Paillier cryptosystem, along with the introduction of the offline and online mechanism to speed up computations. Subsequently, we elaborate on four secure computation protocols based on the proposed threshold Paillier cryptosystem.
The analysis of correctness and security is presented in Section \ref{Section_6}, and experimental evaluations are executed in Section \ref{Section_7}.
Finally, this paper is concluded in Section \ref{Section_8}.

\section{Related Work}\label{Section_2}
Secure outsourced computation is a powerful tool that allows users with limited storage and computation capabilities to outsource their data and computations over data to cloud in a secure manner. 
To enhance the security of data stored in the cloud, numerous solutions for secure outsourced computation have been proposed. 

Rahulamathavan \textit{et al.} \cite{rahulamathavan2013privacy} proposed a privacy-preserving approach for outsourcing support vector machine data classification to the cloud, which is based on a single-server architecture.
In their work \cite{rahulamathavan2013privacy}, the operations over encrypted data are rely on Paillier cryptosystem \cite{paillier} and secure two-party computation.
In the work \cite{zhu2013privacy}, a secure outsourced  approach for logistic regression in cloud is proposed, which is also based on Paillier cryptosystem and single-server architecture.
Despite the utilization of a powerful cloud server in the aforementioned work, the burden on the client is not truly alleviated due to  the execution of some interactions between client and server.

Twin-server architecture emerges as a more practical solution for secure outsourced computation, which significantly reduces the burden of client (or data user).
The schemes proposed in \cite{erkin2012generating}, \cite{elmehdwi2014secure}, \cite{chun2014outsourceable}, \cite{samanthula2014privacy} exploited twin-server architecture and Paillier cryptosystem to implement secure outsourced computation.
Wang \textit{et al.} \cite{wang2014tale} implemented a secure addition using the twin-server architecture and the ElGamal-based proxy re-encryption with multiplicatively homomorphism.
Feng \textit{et al.} \cite{feng2018privacy} leveraged Paillier cryptosystem and twin-server architecture, proposing a secure integer division protocol (SD) and a secure integer square root protocol (SSR).
Cui \textit{et al.} \cite{cui2020svknn} proposed a secure division computation protocol (SDC) that leverages random numbers to conceal the real value of dividend and divisor. 
However, in all of the aforementioned work, a server with a private key is introduced, leading to a single point of security failure.
Furthermore, the above work fails to access to intermediate result, making it unsuitable in some settings such as privacy-preserving person re-identification \cite{zhao2023identifiable}.

To mitigate the risk of a single point of security failure and enable access to intermediate result, extensive research has been conducted.
The work in \cite{liu2016privacy} and \cite{liu2016efficient_rational_numbers} proposed secure outsourced computation solutions by exploiting threshold Paillier cryptosystem.
Specifically, in the work \cite{liu2016privacy}, the private key of Paillier cryptosystem is split into two parts and distributed to two servers.
Consequently, the ciphertext can be decrypted collaboratively by the two servers.
In the work \cite{liu2016efficient_rational_numbers}, the private key is split into multiple partially private keys held by multiple servers, and a ciphertext can be decrypted by a threshold number of servers holding different partially private keys.
However, both \cite{liu2016privacy} and \cite{liu2016efficient_rational_numbers} introduce a trusted third party to distribute and manage the private keys.

To overcome all the aforementioned weaknesses, a toolkit for secure outsourced computation on integer named SOCI is proposed in \cite{zhao2022soci}, which is based on twin-server architecture and threshold Paillier cryptosystem.
In addition to supporting additive homomorphism and scalar multiplication homomorphism, SOCI enables secure outsourced computation for four types, including secure multiplication, secure comparison, secure sign bit-acquisition, and secure division.
Compared to the protocols in the integer calculation toolkit proposed by \cite{liu2016privacy}, the protocols of SOCI are more efficient. 
However, the computation costs and communication costs of SOCI are still relatively high.

\section{Preliminaries} \label{Section_3}
Ma \textit{et al.} \cite{ma2021optimized} proposed a Paillier cryptosystem with fast encryption and decryption.
In the rest of this paper, we refer to it as  FastPai for its fast encryption and decryption.
FastPai is comprised of the following components. 
$n(\kappa)$ and $l(\kappa)$ refer to the bit length of $N$ and private key, respectively.

\subsubsection{\textbf{N Generation} \texttt{(NGen)}}
FastPai calls \texttt{NGen} to generate the modulus $N$ for the Paillier cryptosystem. Specifically, \texttt{NGen} takes a security parameter $\kappa$ as input and outputs $(N,P,Q,p,q)$. 

The execution of \texttt{NGen} proceeds as follows. 
\begin{enumerate}[  (i)]
    \item
    Randomly select $\frac{l(\kappa)}{2}$-bit odd primes $p,q$.
    \item
    Randomly select ($\frac{n(\kappa) - l(\kappa)}{2}-1)$-bit odd integers $p',q'$. 
    \item
    Compute $P=2pp' + 1$ and $Q=2qq'+1$.
    \item
    If $p,q,p',q'$ are not co-prime, or if P or Q is not a prime, then go back to step (i).
    \item
    Compute $N = PQ$, and output $(N,P,Q,p,q)$.
\end{enumerate}

\subsubsection{\textbf{Key generation} (\texttt{KeyGen})}\label{KeyGen of optimized Paillier}
\texttt{KeyGen} generates a private key $sk$ and a public key $pk$ based on a given parameter $\kappa$.
It firstly calls \texttt{NGen} to obtain $(N,P,Q,p,q)$.
Subsequently, it computes $\alpha = pq$ and $\beta = (P-1)(Q-1)/(4pq)$.
Next, it computes $h = -y^{2\beta} (\!\!\! \mod N)$, where $y$ is a number chosen from $\mathbb{Z}_{N}^{*}$ uniformly and randomly.
Finally, it outputs $pk = (N,h)$ and $sk = \alpha$.
\subsubsection{\textbf{Encryption} \texttt{(Enc)}}\label{Enc of optimized Paillier}
\texttt{Enc} takes a message $m \in {\mathbb{Z}}_{N}$ and a public key $pk = (N,h)$ as input and outputs a ciphertext $c \in {\mathbb{Z}}_{{N}^{2}}$, which is defined as follows.
\begin{align}\label{encryption____________________________}
&c \leftarrow \texttt{Enc}(pk,m)={(1+N)}^m\cdot {(h^r \!\!\!\! \mod N)}^N  \!\!\!\! \mod N^2. %\nonumber%\\
\end{align}

In Eq.$(\ref{encryption____________________________})$, $r$ is a random number and satisfying $r \leftarrow {\{0,1\}}^{l(\kappa)}$.
In the rest of this paper, we use $\llbracket x \rrbracket $ to represent an encrypted $x$.

\subsubsection{\textbf{Decryption} \texttt{(Dec)}}\label{Dec of optimized Paillier}
\texttt{Dec} takes a ciphertext $c \in {\mathbb{Z}}_{{N}^{2}}$ and a private key $sk = \alpha$ as input and outputs a plaintext message $m \in {\mathbb{Z}}_{N}$,
which is defined as follows.
\begin{align}%\label{decrytion}
&m \leftarrow \texttt{Dec}(sk,c)\nonumber \\
&=(\frac{(c^{2\alpha}\ \!\!\!\! \mod\ N^2)-1}{N}\ \!\!\!\! \mod\ N)\cdot{(2\alpha)}^{-1}\ \!\!\!\! \mod\ N.%\nonumber%\\
\end{align}

\section{System model and threat model}\label{Section_4}

\begin{figure}[!ht]
    \centering
    \resizebox{0.35\textwidth}{!}{    \includegraphics{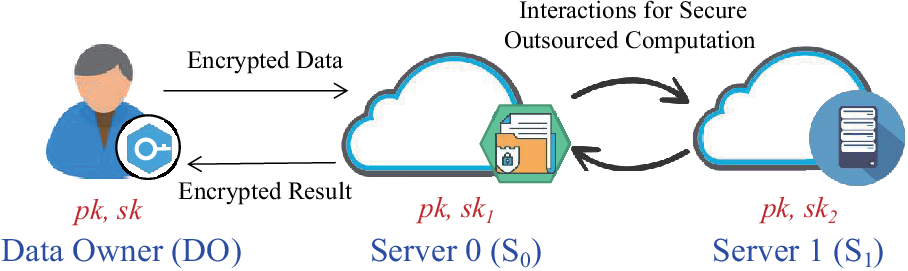}
    }
    \caption{SOCI\textsuperscript{+} system architecture}
    \label{fig:system architecture}
\end{figure}

\subsection{System Model}
As shown in Fig. \ref{fig:system architecture}, SOCI\textsuperscript{+} consists of a Data Owner (DO) and two non-colluding servers ($S_0$ and $S_1$).

\begin{itemize}
    \item \textbf{Data Owner (DO)}.
    DO is responsible for generating the private key and public key of FastPaiTD, and distributing the public key and the partially private keys $sk_1$ and $sk_2$ to $S_0$ and $S_1$, respectively.
    To ensure data security, DO encrypts data with $pk$ and outsources the encrypted data to $S_0$. Subsequently, DO outsources the computations on ciphertext to $S_0$ and $S_1$.
    \item \textbf{Servers}.
    $S_0$ is responsible for the storage and the management of the encrypted data uploaded by DO. 
    Additionally, $S_0$ interacts with $S_1$ to perform the proposed secure outsourced computation protocols.
    $S_1$ only provides computation services and collaborates with $S_0$ to perform the proposed secure outsourced computation protocols.
\end{itemize}

\subsection{Threat Model}
Following the previous work falling in twin-server architecture \cite{zhao2022soci}, \cite{nikolaenko2013privacy}, \cite{mohassel2017secureml}, \cite{chen2021secrec}, SOCI\textsuperscript{+} comprises three entities, DO, $S_0$ and $S_1$. 
DO is regarded as fully trusted. In SOCI\textsuperscript{+}, there is only one type of adversary, which involves $S_0$ and $S_1$, and the adversary attempts to obtain DO's data during execution of secure outsourced computations. 
Similar to the previous the solutions \cite{xie2021achieving,hu2023achieving}, we assume that $S_0$ and $S_1$ are non-colluding.
Moreover, we assume $S_0$ and $S_1$ both are \textit{curious-but-honest}, i.e., both $S_0$ and $S_1$ strictly adhere to the principle of not revealing additional information to each other, except for the necessary information required for performing secure outsourced computations.   

It is practical that assuming $S_0$ and $S_1$ are non-colluding, when $S_0$ and $S_1$ are two different and competitive cloud service providers.
The collusion between $S_0$ and $S_1$ means that they share the private information (e.g., the partially private keys and the random numbers) to each other.
Once $S_0$ leaks information to $S_1$, $S_1$ can leverage the law to punish $S_0$ and further occupies the market share of $S_0$, and vise versa.
For the interest of business, both $S_0$ and $S_1$ will not reveal its private information to each other.

\section{SOCI\textsuperscript{+} Design} \label{Section_5}
\subsection{$(2, 2)$-threshold Paillier cryptosystem (FastPaiTD)}\label{FastPaiTD}
Inspired by the work \cite{zhao2022soci} and \cite{liu2016privacy}, we propose FastPaiTD, a novel $(2,2)$-threshold Paillier cryptosystem, which is based on FastPai \cite{ma2021optimized}.
FastPaiTD encompasses the operations of \texttt{NGen}, \texttt{KeyGen}, \texttt{Enc}, and \texttt{Dec} from FastPai.

Previous works such as the PaillierTD \cite{lysyanskaya2001adaptive} adopted by SOCI \cite{zhao2022soci} and the PCPD in POCF \cite{liu2016privacy} split the private key (e.g., $sk = \lambda$) into two partially private keys $sk_1$ and $sk_2$. In contrast to these methods, we split FastPai's double private key (e.g., $2sk = 2\alpha$) into two partially private keys $sk_1$ and $sk_2$, s.t., $sk_1 + sk_2=0\!\! \mod 2\alpha$ and $sk_1 + sk_2 = 1 \!\! \mod N$. 
%It should be noted that we have modified the \textbf{keygen} in \textbf{FastPai}. 
In the output of \texttt{keygen} in FastPai, $N$ is an odd number that satisfies $\gcd(2\alpha, N)=1$. 
To hold $sk_1 + sk_2 = 0 \!\! \mod 2\alpha$ and $sk_1 + sk_2 = 1 \!\! \mod N$ at the same time, we can apply the Chinese remainder theorem \cite{pei1996chinese} to calculate $\delta = sk_1 + sk_2 = (2\alpha)\cdot ((2\alpha)^{-1} \!\! \mod N)\ \!\! \mod (2\alpha \cdot N$). We can randomly set $sk_1$ as a
$\sigma$-bit (e.g., $\sigma = 128$) number, and set $sk_2 = ( (2\alpha)^{-1}\!\! \mod N) \cdot (2\alpha) - sk_1 + \eta \cdot 2\alpha \cdot N$, where $\eta \ge 0$.  
The splitting operation of the private key should be performed in the \texttt{keygen} phase.

In addition to the fundamental components of FastPai, FastPaiTD incorporates \texttt{PDec} and \texttt{TDec} operations. These supplementary operations significantly enhance the flexibility of FastPaiTD, making it a practical tool for secure outsourced computation.

\textbf{Partial Decryption} \texttt{(PDec)}: This operation enables a party to partially decrypt the ciphertext without revealing the original message. 
\texttt{PDec} takes a ciphertext $c \in {\mathbb{Z}}_{{N}^{2}}$ and a partially private key $sk_i$ ($i \in \{1,2\}$) as input, and outputs a ciphertext $M_{i} \in {\mathbb{Z}}_{{N}^{2}}$. The partial decryption process is defined as follows.
\begin{align}
    M_{i}\leftarrow \texttt{PDec}(sk_{i},c)=c^{sk_{i} } \!\!\! \mod N^{2}.
\end{align}

\begin{figure*}[!ht]
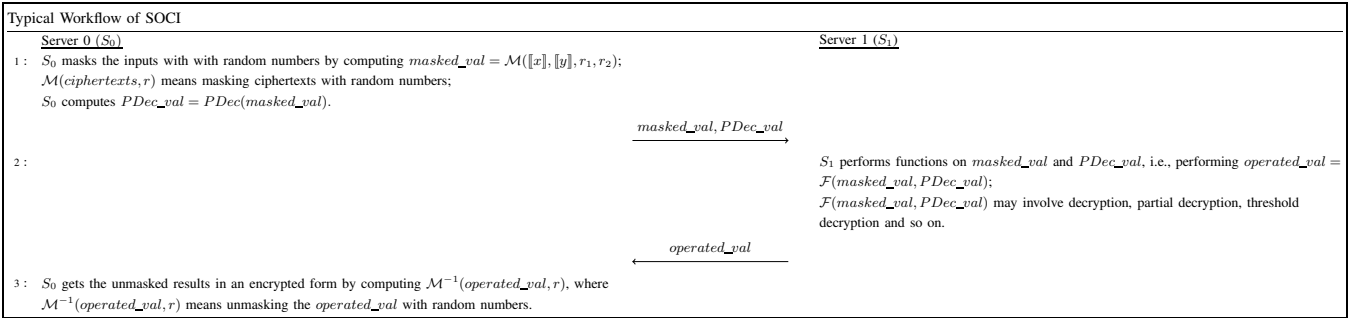

    % \centering
    \resizebox{0.60\textwidth}{!}{
    \begin{minipage}{1.0\textwidth}
    \procb[linenumbering, skipfirstln,mode=text]{Typical Workflow of SOCI}{
        \uline{Server 0 $(S_0)$}  \< \< \uline{Server 1 ($S_1$)} \\ 
        $S_0$ masks the inputs with with random numbers by computing $masked\_val = \mathcal{M}(\llbracket x \rrbracket, \llbracket y \rrbracket, r_1, r_2)$; \pcskipln \\
         % \pcskipln \\ 
        $\mathcal{M}(ciphertexts, r)$ means masking ciphertexts with random numbers;\pcskipln \\ 
         % \pcskipln \\
        $S_0$ computes $PDec\_val=PDec(masked\_val).$ \pcskipln \\ 
        \< \sendmessageright*{masked\_val,PDec\_val} \<  \\
        \< \< $S_1$ performs functions on $masked\_val$ and $PDec\_val$, i.e., performing $operated\_val = $\pcskipln \\ 
        % \< \<  \pcskipln \\
        \< \< $\mathcal{F}(masked\_val, PDec\_val)$; \pcskipln \\
        \< \< $\mathcal{F}(masked\_val, PDec\_val)$ may involve decryption, partial decryption, threshold\pcskipln \\
        \< \< decryption and so on. \pcskipln \\
        % \< \<  \pcskipln \\
        \< \sendmessageleft*{operated\_val}\\
        $S_0$ gets the unmasked results in an encrypted form by computing $\mathcal{M}^{-1}(operated\_val, r)$, where\pcskipln\\
         % \pcskipln\\
        $\mathcal{M}^{-1}(operated\_val, r)$ means unmasking the $operated\_val$ with random numbers.%\pcskipln\\
    }
    \end{minipage}
    }
    \caption{Typical Workflow of SOCI}
    \label{alg:Top Workflow of SOCI}
\end{figure*}

\textbf{Threshold Decryption} \texttt{(TDec)}:
This operation enables two authorized parties to collaboratively decrypt the ciphertext and obtain the original message without knowing the private key $sk$.
\texttt{TDec} takes the results of partial decryption $M_{1}$ and $M_{2}$ as input, and outputs a plaintext $m \in {\mathbb{Z}}_{N}$.
The threshold decryption process is defined  as follows.
\begin{align}
    m\leftarrow \texttt{TDec}(M_{1},M_{2})=\frac{(M_{1} \cdot M_{2}  \!\!\! \mod N^{2})-1}{N} \!\!\!\! \mod N.
\end{align}

\textbf{Remark}. Similar to the PaillierTD \cite{lysyanskaya2001adaptive} adopted by SOCI \cite{zhao2022soci}, 
our $(2,2)$-threshold Paillier cryptosystem (FastPaiTD) supports additive homomorphism and scalar-multiplication homomorphism:
\begin{itemize}
    \item 
    $\texttt{Enc}(pk,m_1)\cdot \texttt{Enc}(pk,m_2) = \texttt{Enc}(pk,m_1+m_2)$.
    \item
    ${\texttt{Enc}(pk,m)}^r = \texttt{Enc}(pk,r\cdot m)$, where $r$ is a constant. When $r = N-1$, it holds ${\texttt{Enc}(pk,m)}^r = \texttt{Enc}(pk,-m)$.
\end{itemize}

Besides, to enable FastPaiTD supporting the operations on negative integers, we perform a conversion on the negative number $m$ as $m = N - \left | m \right | $. 
Specifically, we take the message spaces $[0,\frac{N}{2}]$ and $[\frac{N}{2}+1, N-1]$ for non-negative numbers and negative numbers, respectively.

\subsection{Offline and online mechanism}
As shown in Fig. \ref{alg:Top Workflow of SOCI}, to hide the real values of inputs, SOCI masks the inputs with random numbers. 
To securely and correctly obtain the results, the protocols in SOCI involve a large amount of encryption for random numbers.
To avoid the expensive encryption overhead during executing the secure outsourced computation protocols, we propose an offline and online mechanism for SOCI\textsuperscript{+} as detailed below.
\subsubsection{Offline Phase}
In contrast to SOCI, we pre-encrypt the random numbers and some constants in the offline phase, such as $r_1$, $r_2$, $-r_1 \cdot r_2$, $0$ and $1$, thereby avoiding to encrypt them in the online phase, which alleviates the computation costs for the secure outsourced computation protocols.
Specifically, in the offline phase, we separately establish a tuple for $S_0$ and $S_1$, and denote them as \textit{$tuple_{S_0}$} and \textit{$tuple_{S_1}$}, respectively.
\textit{$tuple_{S_0}$} is consist of $r_1$, $r_2$, $\llbracket r_1 
\rrbracket$, $\llbracket r_2 \rrbracket$, $\llbracket -r_1 \cdot r_2 \rrbracket $, $r_3$, $r_4$, $\llbracket r_3+r_4 \rrbracket $, $\llbracket r_4 \rrbracket $, $\llbracket 0 \rrbracket $ and $\llbracket 1 \rrbracket $. The elements in \textit{$tuple_{S_0}$} satisfy the following properties.
\begin{itemize}
    \label{random_number_scope}
    \item
    $r_1$, $r_2$ $\leftarrow$ $\{0,1\}^\sigma$ (e.g., $\sigma = 128$).
    \item
    $r_3 \leftarrow {\{0,1\}}^{\sigma}\backslash \{0\}$ (e.g., $\sigma=128$). 
    \item
    $r_4$ is a random number, s.t., $r_4 \le \frac{N}{2} $ and $r_3+r_4>\frac{N}{2}$.
\end{itemize}

\textit{$tuple_{S_1}$} is consist of $\llbracket 0 \rrbracket $ and $\llbracket 1 \rrbracket $.
Compared to SOCI, $S_0$ and $S_1$ in SOCI\textsuperscript{+} have a simplified process where they only need to extract a ciphertext from their tuples and refresh it atfer usage when it comes to encryption of random numbers and some constants.

However, there is still a number needed to be encrypted in the online phase when we adopt the above mechanism in our SOCI\textsuperscript{+}.
To speed up the encryption, we can construct a pre-computation table in the offline phase.
Moreover, the \texttt{Enc} in FastPai has another equivalent form, as shown below.
\begin{align}
    c\leftarrow \texttt{Enc}(pk,m) = (1+m\cdot N)\cdot {(h^N \!\!\!\! \mod N^2)}^r \!\!\!\! \mod N^2.    
\end{align}

The \texttt{Enc} involves a constant $h^N \!\!\!\mod N^2$, hence we can pre-compute this constant to speed up the \texttt{Enc}. Besides, the \texttt{Enc} in FastPai involves a fixed-base modular exponentiation as below.
\begin{align}
    {(h^N\!\!\! \mod\ N^2)}^r \!\!\! \mod N^2.
\end{align}

Therefore, constructing a pre-computation table can optimize the efficiency of the \texttt{Enc}. Ma \textit{et al.} \cite{ma2021optimized} presented the method of constructing a pre-computation table, which is detailed as follows.

\begin{figure*}[!ht]
    % \centering
    \scalebox{0.60}{
        \begin{minipage}{1.0\textwidth}
        \procb[linenumbering, skipfirstln, mode=text]{Secure Multiplication (\texttt{SMUL}): \texttt{SMUL}($\llbracket x \rrbracket, \llbracket y \rrbracket) \rightarrow \llbracket x\cdot y \rrbracket$}{
            \uline{Server 0 $(S_0)$}  \< \< \uline{Server 1 ($S_1$)} \pcskipln \\
            %[][\hline] \pcskipln 
           \pclinecomment{\texttt{Offline Phase}} \pcskipln \\
            $S_0$ constructs a $tuple_{S_0}$, which is consist of $r_1$, $r_2$, $\llbracket r_1 \rrbracket$, $\llbracket r_2 \rrbracket$, $\llbracket -r_1 \cdot r_2 \rrbracket$, $r_3$, $r_4$, $\llbracket r_3+r_4 \rrbracket$,\pcskipln \\ 
             % \pcskipln \\
            $\llbracket r_4 \rrbracket$, $\llbracket 0 \rrbracket$ and $\llbracket 1 \rrbracket$. \pcskipln\\
            \< \< $S_1$ constructs a pre-computation table for speeding up the encryption that cannot pre-compute.\pcskipln \\
            % \< \<  \pcskipln \\ 
            % [][\hline] 
            \pclinecomment{\texttt{Online Phase}}  \\
            $S_0$ inputs $\llbracket x \rrbracket$ and $\llbracket y \rrbracket$, where $x,y \in [-2^l,2^l]$; \pcskipln \\
            $S_0$ extracts $r_1$, $r_2$, $\llbracket r_1 \rrbracket$, $\llbracket r_2 \rrbracket$, and $\llbracket -r_1\cdot r_2 \rrbracket$ from $tuple_{S_0}$;\pcskipln \\
             % \pcskipln \\ 
            $S_0$ refreshes $\llbracket r_1 \rrbracket $, $\llbracket r_2 \rrbracket $ and $\llbracket -r_1\cdot r_2 \rrbracket $ in $tuple_{S_0}$, using the method mentioned in \ref{refresh}; \pcskipln \\
            $S_0$ computes $X = \llbracket x \rrbracket \cdot \llbracket r_1 \rrbracket $ and $Y = \llbracket y \rrbracket \cdot \llbracket r_2 \rrbracket $; \pcskipln \\
            $S_0$ computes $C = X^{L}\cdot Y$ and $C_1 \leftarrow $\texttt{PDec}$(sk_1,C)$, where $L$ is a constant satisfying\pcskipln \\
            $L\ge 2^{\sigma +2}$. \pcskipln \\
            \< \sendmessageright*{C,C_1} \<  \\
            \< \< $S_1$ computes $C_2 \leftarrow$ \texttt{PDec} $(sk_2,C)$ and $(L\cdot (x+r_1)+y+r_2)\leftarrow $ \texttt{TDec}$(C_1,C_2)$; \pcskipln \\
            % \< \<  \pcskipln \\
            \< \< $S_1$ sets $(x+r_1) = \left \lfloor (L\cdot (x+r_1)+y+r_2)/L \right \rfloor$; \pcskipln \\
            % \< \< \pcskipln \\
            \< \< $S_1$ computes $(y+r_2) = (L\cdot (x+r_1)+y+r_2)\!\!\! \mod L$ and $\llbracket (x+r_1)\cdot (y+r_2) \rrbracket  \leftarrow \texttt{Enc}$ \pcskipln \\
            \< \< $(pk,(x+r_1)\cdot (y+r_2))$. \pcskipln \\
            \< \sendmessageleft*{\llbracket (x+r_1)\cdot (y+r_2) \rrbracket} \< \\
            $S_0$ computes $\llbracket -r_2\cdot x \rrbracket  = {\llbracket x \rrbracket }^{-r_2}$ and  $\llbracket -r_1\cdot y \rrbracket  = {\llbracket y \rrbracket }^{-r_1}$; \pcskipln \\
             % \pcskipln \\
            $S_0$ gets $\llbracket x\cdot y \rrbracket  = \llbracket (x+r_1)\cdot (y+r_2) \rrbracket  \cdot \llbracket -r_2x \rrbracket \cdot \llbracket -r_1y \rrbracket \cdot \llbracket -r_1\cdot r_2 \rrbracket $. 
            % \pcskipln \\
            % \key_A \gets Y^x 
        }
        \end{minipage}
    }
    \caption{Secure Multiplication (\texttt{SMUL})}
    \label{alg:algorithm-label-SMUL-FREED}
\end{figure*}

To compute $y = a^x$, where $a$ is a fixed base, we can pre-compute the powers of $a$ so that turn the modular exponentiation into modular multiplication since modular multiplication is more efficient than modular exponentiation. Specifically, we let $x =  {\textstyle \sum_{i=0}^{\left \lceil len/b \right \rceil - 1 }}x_i \cdot 2^{ib} $, where $len$ is the bit length of $x$ and $x_i$ is the $i$-th $b$-bit block. Note that the last block $x_{\left \lceil len/b \right \rceil - 1 }$ may be less than $b$ bit. We can calculate $y = a^x$ by the following equation.
\begin{align}
    y = a^x =a^{{\textstyle \sum_{i=0}^{\left \lceil len/b \right \rceil - 1 }}x_i \cdot 2^{ib}}= {\textstyle \prod_{i=0}^{\left \lceil len/b \right \rceil - 1 }} {(a^{2^{ib}})}^{x_i}.
\end{align}

Therefore, we can build a two-dimensional pre-computation table with $\left \lceil len/b \right \rceil$ rows and $2^b$ columns, and the index of rows and columns start from 0. The element in row i and column j is ${(a^{2^{ib}})}^{j}$, where $i \in [0,\left \lceil len/b \right \rceil-1]$ and $j \in [0,2^b-1]$. The table has $\left \lceil len/b \right \rceil \cdot 2^b$ elements and every element belongs to ${\mathbb{Z}}_{{N}^{2}}$, hence the table size is $\left \lceil len/b \right \rceil \cdot 2^b \cdot (2n)$ bits.

\subsubsection{Online Phase}
In SOCI\textsuperscript{+}, $S_0$ and $S_1$ construct the \texttt{$tuple_{S_0}$}, \textit{$tuple_{S_1}$} and pre-computation table in the offline phase, and utilize the \texttt{$tuple_{S_0}$}, \textit{$tuple_{S_1}$} and pre-computation table when perform secure outsourced computation protocols in the online phase.

During the execution of secure outsourced computation protocols, 
SOCI performs encryption operations on random numbers and some constants, whereas SOCI\textsuperscript{+} 
extracts the pre-computed encryption values from tuples and hence reduces a large amount of encryption operations.
In the proposed FastPaiTD, multiplying a ciphertext \texttt{Enc}$(pk,m)$ by a ciphertext \texttt{Enc}$(pk,0)$ produces a new ciphertext \texttt{Enc}$(pk,m+0)$. Although \texttt{Enc}$(pk,m)$ and \texttt{Enc}$(pk,m+0)$ are not identical, their decrypted results are identical. Consequently, after utilizing the ciphertexts in tuples, $S_0$ and $S_1$ refresh the ciphertexts in their tuples by multiplying $\llbracket 0 \rrbracket $. It should be noted that the $\llbracket 0 \rrbracket $ is also included in their tuples. By refreshing the ciphertext, even if the plaintext remains unchanged, the corresponding ciphertext changes.
This refresh process creates the illusion that all random numbers and constants are re-encrypted, providing security while reducing computation cost. \label{refresh}
Moreover, $S_0$ and $S_1$ can utilize the pre-computation table to expedite the encryption process when encrypting messages other than the aforementioned random numbers and constants.

\begin{figure*}[!ht]
    % \centering
    \scalebox{0.60}{
    \begin{minipage}{1.0\textwidth}
    \procb[linenumbering, skipfirstln, mode=text]{Secure Comparison (\texttt{SCMP}): \texttt{SCMP}($\llbracket x \rrbracket, \llbracket y \rrbracket) \rightarrow \llbracket \mu \rrbracket$}{
        \uline{Server 0 $(S_0)$}  \< \< \t\t\t\t \uline{Server 1 ($S_1$)}   \pcskipln \\
        % [][\hline] \pcskipln 
       \pclinecomment{\texttt{Offline Phase}} \pcskipln \\
        $S_0$ constructs a $tuple_{S_0}$, which is consist of $r_1$, $r_2$, $\llbracket r_1 \rrbracket$, $\llbracket r_2 \rrbracket$, $\llbracket -r_1 \cdot r_2 \rrbracket$, $r_3$, $r_4$,\pcskipln \\ 
         % \pcskipln \\
        $\llbracket r_3+r_4 \rrbracket$, $\llbracket r_4 \rrbracket$, $\llbracket 0 \rrbracket$ and $\llbracket 1 \rrbracket$. \pcskipln\\
        \< \< \t\t\t\t $S_1$ constructs a $tuple_{S_1}$, which is consist of $\llbracket 0 \rrbracket $ and $\llbracket 1 \rrbracket $. \pcskipln\\
        % \< \< . \pcskipln\\
        % [][\hline] 
        \pclinecomment{\texttt{Online Phase}}  \\
        $S_0$ inputs $\llbracket x \rrbracket$ and $\llbracket y \rrbracket$, where $x,y \in [-2^l,2^l]$; \pcskipln \\
        $S_0$ extracts $r_3$, $r_4$, $\llbracket r_3+r_4\rrbracket $ and $\llbracket r_4 \rrbracket $ from $tuple_{S_0}$, then labels them as $r_1$, $r_2$, $\llbracket r_1+r_2\rrbracket $\pcskipln\\
          % \pcskipln\\
        and $\llbracket r_2\rrbracket $, respectively; \pcskipln\\
        $S_0$ refreshes $\llbracket r_3+r_4\rrbracket $ and $\llbracket r_4\rrbracket $ in $tuple_{S_0}$, using the method mentioned in \ref{refresh};\pcskipln\\
         % \pcskipln\\
        $S_0$ randomly chooses a number $\pi$, where $\pi \leftarrow \{0,1\}$;\pcskipln\\
         % \pcskipln\\ 
        If $\pi = 0$, $S_0$ calculates $D={(\llbracket x \rrbracket \cdot {\llbracket y \rrbracket }^{N-1})}^{r_1}\cdot \llbracket r_1+r_2 \rrbracket $,
        otherwise, $S_0$ calculates $D=$\pcskipln\\
         % \pcskipln\\
        $ {(\llbracket y \rrbracket \cdot {\llbracket x\rrbracket }^{N-1})}^{r_1}\cdot \llbracket r_2 \rrbracket $; \pcskipln\\
        $S_0$ computes $D_1 \leftarrow \texttt{PDec}(sk_1,D)$. \pcskipln\\
        \< \sendmessageright*{D,D_1} \<  \\
        \< \< \t\t\t\t $S_1$ computes $D_2 \! \leftarrow \! \texttt{PDec}(sk_2,D)$ and $d \! \leftarrow \! \texttt{TDec}(D_1,D_2)$;\pcskipln\\
        % \<\<  \pcskipln\\
        \< \< \t\t\t\t $S_1$ extracts $\llbracket 0 \rrbracket $ and $\llbracket 1 \rrbracket $ from $tuple_{S_1}$, then refreshes $\llbracket 0 \rrbracket $ and $\llbracket 1 \rrbracket $ in $tuple_{S_1}$, using the  \pcskipln\\
        \<\< \t\t\t\t method mentioned in \ref{refresh}; \pcskipln\\
        % \< \<  \pcskipln\\
        \< \< \t\t\t\t If $d >\frac{N}{2}$, $S_1$ sets $\llbracket {\mu}_0 \rrbracket$ = $ \llbracket 0 \rrbracket $, otherwise, $S_1$ sets $\llbracket {\mu}_0 \rrbracket = \llbracket 1 \rrbracket $.  \pcskipln\\
        % \< \<   \pcskipln\\
        \< \sendmessageleft*{\llbracket {\mu}_0 \rrbracket} \< \\
        If $\pi = 0$, $S_0$ sets $\llbracket \mu \rrbracket  = \llbracket {\mu}_0 \rrbracket $; \pcskipln\\
        Conversely, if $\pi = 1$, $S_0$ extracts $\llbracket 1 \rrbracket $ from $tuple_{S_0}$, and then refreshes $\llbracket 1 \rrbracket $ in $tuple_{S_0}$, \pcskipln\\ 
         % \pcskipln\\
        using the method mentioned in  \ref{refresh}, finally $S_0$ sets $\llbracket \mu \rrbracket  = \llbracket 1 \rrbracket \cdot {\llbracket {\mu}_0\rrbracket }^{N-1}$.
        % \pcskipln\\
    }
    \end{minipage}
    }
    \caption{Secure Comparison (\texttt{SCMP})}
    \label{alg:algorithm-label-SCMP}
\end{figure*}

\subsection{Secure Multiplication Protocol (\texttt{SMUL})}

% In this subsection, we re-design the protocols (\texttt{SMUL}, \texttt{SCMP}, \texttt{SSBA}, and \texttt{SDIV}) in SOCI, and propose an enhanced toolkit for secure outsourced computation on integers named SOCI\textsuperscript{+}.

FREED \cite{zhao2022freed} proposed a \texttt{SMUL} protocol which is more efficient than the one in SOCI and with the same input and output as SOCI. 
Same as SOCI, FREED splits the private key of Paillier cryptosystem into two parts and achieves \texttt{SMUL} through the interaction between the two servers.
In this paper, we re-design the \texttt{SMUL} in FREED by incorporating the proposed FastPaiTD and the offline and online computation mechanism.

% To construct SOCI\textsuperscript{+}, we adopt the \texttt{SMUL} in FREED, the \texttt{SCMP}, \texttt{SSBA}, \texttt{SDIV} in SOCI, and re-design all the protocols by leveraging the proposed FastPaiTD and the offline and online computation mechanism. 
% Now, we introduce the protocols in our SOCI\textsuperscript{+}.
% The protocols in SOCI\textsuperscript{+} have $x,y \in [-2^l,2^l]$ for the \texttt{SMUL}, \texttt{SCMP} and \texttt{SSBA}, where $l$ denotes the bit length of a message (e.g. $l=32$).
% For \texttt{SDIV}, SOCI\textsuperscript{+} has $x \in [0,2^l]$ and $y \in (0,2^l]$.

Given $\llbracket x \rrbracket $ and $\llbracket y \rrbracket $ as input, where  $x,y \in [-2^l,2^l]$, $S_0$ and $S_1$ collaboratively compute $\llbracket x \cdot y \rrbracket  \leftarrow \texttt{SMUL}(\llbracket x \rrbracket, \llbracket y \rrbracket )$ as output.
It should be noted that the input is held by $S_0$ and only $S_0$ has the access to the output.
When describing \texttt{SMUL} in Fig. \ref{alg:algorithm-label-SMUL-FREED}, we omit the input and output for conciseness.
As shown in Fig. \ref{alg:algorithm-label-SMUL-FREED}, \texttt{SMUL} has two phase, i.e., offline phase and online phase.
In the offline phase, $S_0$ constructs a \textit{$tuple_{S_0}$} which is consist of $r_1$, $r_2$, $\llbracket r_1 
\rrbracket$, $\llbracket r_2 \rrbracket$, $\llbracket -r_1 \cdot r_2 \rrbracket $, $r_3$, $r_4$, $\llbracket r_3+r_4 \rrbracket $, $\llbracket r_4 \rrbracket $, $\llbracket 0 \rrbracket $ and $\llbracket 1 \rrbracket $ (how to choose these random numbers is elaborated in \ref{random_number_scope}).
Meanwhile, $S_1$ constructs a pre-computation table for speeding up the encryption that cannot pre-compute.
The online phase of \texttt{SMUL} comprises three steps as detailed below.

\begin{figure*}[!ht]
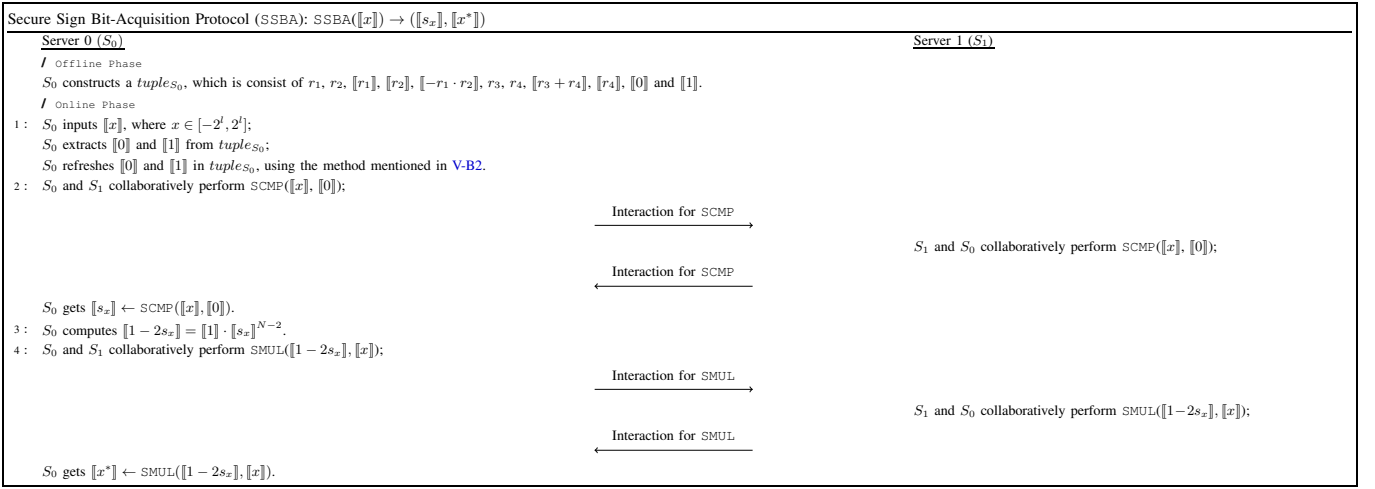

    % \centering
    \scalebox{0.60}{
    \begin{minipage}{1.0\textwidth}
    \procb[linenumbering, skipfirstln, mode=text]{Secure Sign Bit-Acquisition Protocol (\texttt{SSBA}): \texttt{SSBA}($\llbracket x \rrbracket) \rightarrow (\llbracket s_x \rrbracket,\llbracket x^* \rrbracket)$}{
        \uline{Server 0 $(S_0)$}  \< \< \t\t\t\t\t\t\t\t\t \uline{Server 1 ($S_1$)} \pcskipln \\
        % [][\hline \hline] \pcskipln \\
        % [][\hline] \pcskipln  
        % \pcintertext[dotted]{Some D i v i s i o n } \\
       \pclinecomment{\texttt{Offline Phase}} \pcskipln \\
        $S_0$ constructs a $tuple_{S_0}$, which is consist of $r_1$, $r_2$, $\llbracket r_1 \rrbracket$, $\llbracket r_2 \rrbracket$, $\llbracket -r_1 \cdot r_2 \rrbracket$, $r_3$, $r_4$, $\llbracket r_3+r_4 \rrbracket$, $\llbracket r_4 \rrbracket$, $\llbracket 0 \rrbracket$ and $\llbracket 1 \rrbracket$.\pcskipln \\ 
         % \pcskipln \\
         % \pcskipln\\
        % [][\hline] 
        \pclinecomment{\texttt{Online Phase}}  \\
        $S_0$ inputs $\llbracket x \rrbracket$, where $x \in [-2^l,2^l]$; \pcskipln \\
        $S_0$ extracts $\llbracket 0 \rrbracket $ and $\llbracket 1 \rrbracket $ from $tuple_{S_0}$; \pcskipln\\
        $S_0$ refreshes $\llbracket 0 \rrbracket $ and 
        $\llbracket 1 \rrbracket $ in $tuple_{S_0}$, using the method mentioned in \ref{refresh}.\\
        $S_0$ and $S_1$ collaboratively perform \texttt{SCMP}($\llbracket x\rrbracket$, $\llbracket 0\rrbracket$); \pcskipln \\
        \t\t\t\t\t\t\t\t\t\t\t\t\t\t\t\t\t\t\t\t\t\t\t\t\t\t\t\t\t\t\t\t\t\t\t\t\t\t \sendmessageright*{\text{Interaction for \texttt{SCMP}}} \<  \pcskipln \\
        \<\<\t\t\t\t\t\t\t\t\t $S_1$ and $S_0$ collaboratively perform \texttt{SCMP}($\llbracket x\rrbracket$, $\llbracket 0\rrbracket$); \t\t\t\t\t\t\t\t\t \pcskipln \\ 
        \t\t\t\t\t\t\t\t\t\t\t\t\t\t\t\t\t\t\t\t\t\t\t\t\t\t\t\t\t\t\t\t\t\t\t\t\t\t \sendmessageleft*{\text{Interaction for \texttt{SCMP}}} \< \pcskipln \\
        $S_0$ gets $\llbracket s_{x}\rrbracket  \leftarrow \texttt{SCMP}(\llbracket x\rrbracket ,\llbracket 0\rrbracket )$. \\
        $S_0$ computes $\llbracket 1-2s_{x}\rrbracket  = \llbracket 1\rrbracket \cdot {\llbracket s_{x}\rrbracket }^{N-2}$. \\
        $S_0$ and $S_1$ collaboratively perform \texttt{SMUL}($\llbracket 1-2s_{x}\rrbracket,\llbracket x\rrbracket$);\pcskipln \\
         % \pcskipln \\
        \t\t\t\t\t\t\t\t\t\t\t\t\t\t\t\t\t\t\t\t\t\t\t\t\t\t\t\t\t\t\t\t\t\t\t\t\t\t \sendmessageright*{\text{Interaction for \texttt{SMUL}}} \<  \pcskipln \\
        \<\<\t\t\t\t\t\t\t\t\t $S_1$ and $S_0$ collaboratively perform \texttt{SMUL}($\llbracket 1 \! - \! 2s_{x}\rrbracket, \llbracket x\rrbracket$); \pcskipln \\
        % \< \< \pcskipln \\ 
        \t\t\t\t\t\t\t\t\t\t\t\t\t\t\t\t\t\t\t\t\t\t\t\t\t\t\t\t\t\t\t\t\t\t\t\t\t\t \sendmessageleft*{\text{Interaction for \texttt{SMUL}}} \< \pcskipln \\
        $S_0$ gets $\llbracket x^{*}\rrbracket  \leftarrow  \texttt{SMUL}(\llbracket 1-2s_{x}\rrbracket ,\llbracket x \rrbracket )$.
    }
    \end{minipage}
    }
    \caption{Secure Sign Bit-Acquisition Protocol (\texttt{SSBA})}
    \label{alg:algorithm-label-SSBA}
\end{figure*}

\begin{enumerate}[(1)]
    \item 
    $S_0$ extracts $r_1$, $r_2$, $\llbracket r_1 \rrbracket $, $\llbracket 
    r_2 \rrbracket $, and $\llbracket -r_1\cdot r_2 \rrbracket $ from {$tuple_{S_0}$}.
    Subsequently, $S_0$ refreshes these ciphertexts in {$tuple_{S_0}$},
    and masks x and y through additive homomorphism.
    This is accomplished by computing $X = \llbracket x \rrbracket \cdot \llbracket r_1 \rrbracket $ and $Y = \llbracket y \rrbracket \cdot \llbracket r_2 \rrbracket $. 
    $S_0$ then computes $C = X^{L}\cdot Y$, where $L$ is a constant satisfying $L\ge 2^{\sigma +2}$. 
    After partially decrypting $C$ to obtain $C_1$  by calling \texttt{PDec}, $S_0$ sends $C$ and $C_1$ to $S_1$.
    \item 
    Upon receiving $C$ and $C_1$, $S_1$ calls \texttt{PDec} to partially decrypt $C$, resulting in $C_2$.
    Additionally, $S_1$ obtains $L\cdot (x+r_1)+y+r_2$ by calling \texttt{TDec} with $C_1$ and $C_2$.
    Subsequently, $S_1$ computes $\left \lfloor (L\cdot (x+r_1)+y+r_2)/L \right \rfloor $ and $(L\cdot (x+r_1)+y+r_2) \!\!\!  \mod  L$ to derive the values of  $(x+r_1)$ and $(y+r_2)$, respectively.  Finally, $S_1$ calls \texttt{Enc} to encrypt $(x+r_1)\cdot (y+r_2)$ and sends $\llbracket (x+r_1)\cdot (y+r_2) \rrbracket $ to $S_0$.
    \item 
    As having the knowledge of $\llbracket x \rrbracket $, $\llbracket y \rrbracket $, $r_1$, $r_2$ and $\llbracket -r_1\cdot r_2 \rrbracket $, $S_0$ can computes ${\llbracket x \rrbracket }^{-r_2}$ and ${\llbracket y \rrbracket }^{-r_1}$ to get $\llbracket -r_2\cdot x \rrbracket$ and $\llbracket -r_1\cdot y \rrbracket$,  respectively. Subsequently, $S_0$ computes $\llbracket (x+r_1)\cdot (y+r_2) \rrbracket  \cdot \llbracket -r_2x \rrbracket \cdot \llbracket -r_1y \rrbracket \cdot \llbracket -r_1\cdot r_2 \rrbracket $ to get $\llbracket x\cdot y \rrbracket$. 
\end{enumerate}

\subsection{Secure Comparison Protocol (\texttt{SCMP})}

In this subsection, we re-design the \texttt{SCMP} in SOCI by leveraging the proposed FastPaiTD and the offline and online computation mechanism.

Given $\llbracket x \rrbracket $ and $\llbracket y \rrbracket $ as input, where $x,y \in [-2^l,2^l]$, $S_0$ and $S_1$ collaboratively compute $\llbracket \mu \rrbracket  \leftarrow \texttt{SCMP}(\llbracket x \rrbracket, \llbracket y \rrbracket )$ as output. If $\mu=0$, $x\ge y$, otherwise, $x<y$.
It should be noted that the input is held by $S_0$ and only $S_0$ has the access to the output.
When describing \texttt{SCMP} in Fig. \ref{alg:algorithm-label-SCMP}, we omit the input and output for conciseness.
As shown in Fig. \ref{alg:algorithm-label-SCMP}, the proposed \texttt{SCMP} has offline phase and online phase.
In the offline phase, $S_0$ constructs a \textit{$tuple_{S_0}$}, which is consist of $r_1$, $r_2$, $\llbracket r_1 
\rrbracket$, $\llbracket r_2 \rrbracket$, $\llbracket -r_1 \cdot r_2 \rrbracket $, $r_3$, $r_4$, $\llbracket r_3+r_4 \rrbracket $, $\llbracket r_4 \rrbracket $, $\llbracket 0 \rrbracket $ and $\llbracket 1 \rrbracket $.
Meanwhile, $S_1$ constructs a \textit{$tuple_{S_1}$}, which is consist of $\llbracket 0 \rrbracket $ and $\llbracket 1 \rrbracket $.
The online phase of \texttt{SCMP} consists of three steps as detailed bellow.

\begin{figure*}[!ht]
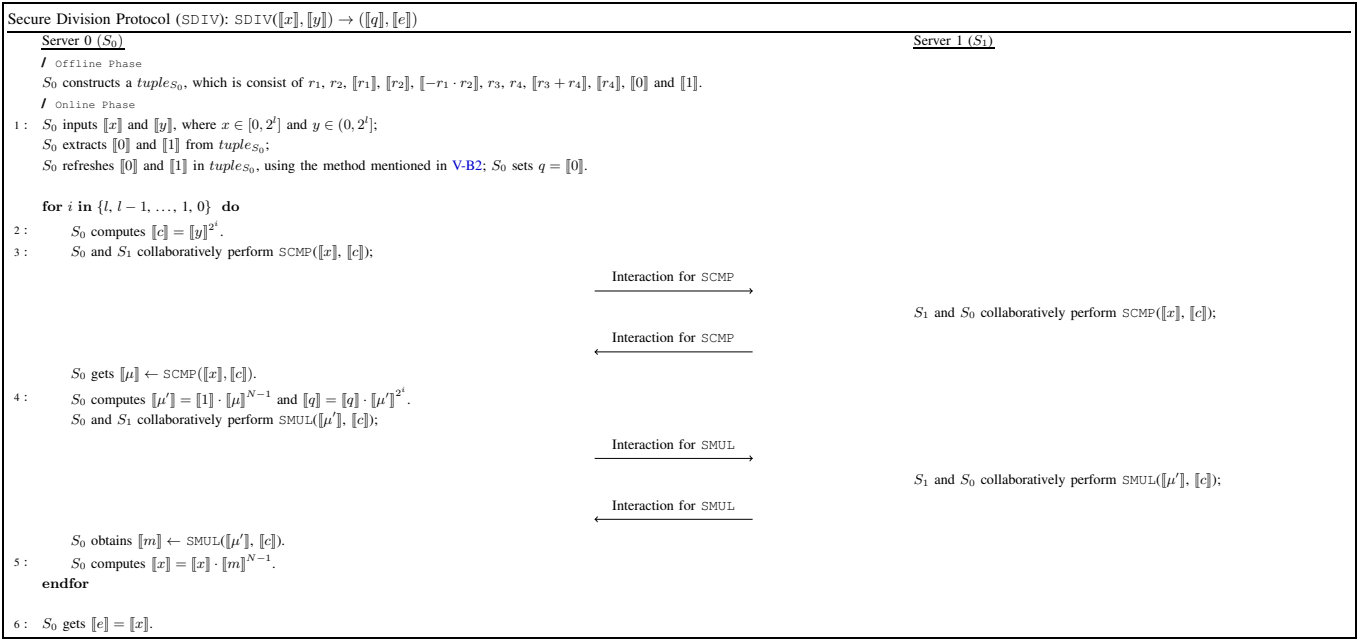

    % \centering
    \scalebox{0.60}{
    \begin{minipage}{1.0\textwidth}
    \procb[linenumbering, skipfirstln, mode=text]{Secure Division Protocol (\texttt{SDIV}): \texttt{SDIV}($\llbracket x \rrbracket, \llbracket y \rrbracket) \rightarrow (\llbracket q \rrbracket,\llbracket e \rrbracket)$}{
        \uline{Server 0 $(S_0)$}  \< \< \t\t\t\t\t\t\t\t\t \uline{Server 1 ($S_1$)} \pcskipln \\
        % [][\hline] \pcskipln
       \pclinecomment{\texttt{Offline Phase}} \pcskipln \\
        $S_0$ constructs a $tuple_{S_0}$, which is consist of $r_1$, $r_2$, $\llbracket r_1 \rrbracket$, $\llbracket r_2 \rrbracket$, $\llbracket -r_1 \cdot r_2 \rrbracket$, $r_3$, $r_4$, $\llbracket r_3+r_4 \rrbracket$, $\llbracket r_4 \rrbracket$, $\llbracket 0 \rrbracket$ and $\llbracket 1 \rrbracket$. \pcskipln \\ 
         % \pcskipln \\
         % \pcskipln\\
        % [][\hline] 
        \pclinecomment{\texttt{Online Phase}}  \\
        $S_0$ inputs $\llbracket x \rrbracket $ and $\llbracket y \rrbracket $, where $x \in [0,2^l]$ and $y \in (0,2^l]$; \pcskipln \\
         % \pcskipln \\
        $S_0$ extracts $\llbracket 0 \rrbracket $ and $\llbracket 1 \rrbracket $ from $tuple_{S_0}$; \pcskipln\\
        $S_0$ refreshes $\llbracket 0 \rrbracket $ and $\llbracket 1 \rrbracket $ in $tuple_{S_0}$, using the method mentioned in \ref{refresh};
        % \pcskipln\\
         % \pcskipln\\
        $S_0$ sets $q=\llbracket 0 \rrbracket$.\pcskipln\\
        \pcskipln\\
        \pcfor $i$ \pcin [] \{$l$, $l-1$, \ldots, 1, 0\} \pcdo  \\
            \t \t $S_0$ computes $\llbracket c \rrbracket  = {\llbracket y \rrbracket }^{2^i}$. \\
            \t \t $S_0$ and $S_1$ collaboratively perform \texttt{SCMP}($\llbracket x\rrbracket$, $\llbracket c\rrbracket$);\pcskipln \\
            % \t \t  \pcskipln \\
            \t\t\t\t\t\t\t\t\t\t\t\t\t\t\t\t\t\t\t\t\t\t\t\t\t\t\t\t\t\t\t\t\t\t\t\t\t\t \sendmessageright*{\text{Interaction for \texttt{SCMP}}} \<  \pcskipln \\
            \<\< \t\t\t\t\t\t\t\t\t $S_1$ and $S_0$ collaboratively perform \texttt{SCMP}($\llbracket x\rrbracket$, $\llbracket c\rrbracket$);\t\t\t\t\t\t\t\t\t \pcskipln \\ 
            \t\t\t\t\t\t\t\t\t\t\t\t\t\t\t\t\t\t\t\t\t\t\t\t\t\t\t\t\t\t\t\t\t\t\t\t\t\t \sendmessageleft*{\text{Interaction for \texttt{SCMP}}} \< \pcskipln \\
            \t\t $S_0$ gets $\llbracket \mu \rrbracket  \leftarrow \texttt{SCMP}(\llbracket x\rrbracket ,\llbracket c\rrbracket )$. \\
            \t \t $S_0$ computes $\llbracket {\mu}' \rrbracket  = \llbracket 1 \rrbracket \cdot {\llbracket \mu \rrbracket }^{N-1}$ and $\llbracket q \rrbracket  = \llbracket q \rrbracket \cdot {\llbracket {\mu}'\rrbracket }^{2^i} $.\pcskipln \\ 
            % \t \t  \\ 
            \t \t $S_0$ and $S_1$ collaboratively perform \texttt{SMUL}($\llbracket {\mu}' \rrbracket$, $\llbracket c \rrbracket$);\pcskipln \\
            % \t \t  \pcskipln \\
            \t\t\t\t\t\t\t\t\t\t\t\t\t\t\t\t\t\t\t\t\t\t\t\t\t\t\t\t\t\t\t\t\t\t\t\t\t\t \sendmessageright*{\text{Interaction for \texttt{SMUL}}} \<  \pcskipln \\
            \<\< \t\t\t\t\t\t\t\t\t $S_1$ and $S_0$ collaboratively perform \texttt{SMUL}($\llbracket {\mu}' \rrbracket$, $\llbracket c \rrbracket$);\t\t\t\t\t\t\t\t\t\pcskipln \\
            \t\t\t\t\t\t\t\t\t\t\t\t\t\t\t\t\t\t\t\t\t\t\t\t\t\t\t\t\t\t\t\t\t\t\t\t\t\t \sendmessageleft*{\text{Interaction for \texttt{SMUL}}} \< \pcskipln \\
            \t \t $S_0$ obtains $\llbracket m \rrbracket  \leftarrow$  \texttt{SMUL}($\llbracket {\mu}' \rrbracket$, $\llbracket c \rrbracket$).\\
            \t \t $S_0$ computes $\llbracket x \rrbracket  = \llbracket x \rrbracket \cdot{\llbracket m \rrbracket }^{N-1}$. \pcskipln \\
        \pcendfor \pcskipln\\ \\
        $S_0$ gets $\llbracket e \rrbracket  = \llbracket x \rrbracket $.
    }
    \end{minipage}
    }
    \caption{Secure Division Protocol (\texttt{SDIV})}
    \label{alg:algorithm-label-SDIV}
\end{figure*}

\begin{enumerate}[(1)]
    \item 
    $S_0$ extracts $r_3$, $r_4$, $\llbracket r_3+r_4\rrbracket $ and $\llbracket r_4 \rrbracket $ from $tuple_{S_0}$, then labels them as
    $r_1$, $r_2$, $\llbracket r_1+r_2\rrbracket $ and $\llbracket r_2\rrbracket $, respectively.
    $S_0$ then refreshes $\llbracket r_3+r_4\rrbracket $ and $\llbracket r_4\rrbracket $ in $tuple_{S_0}$. Next, $S_0$ randomly selects a number $\pi$ from the set $\{0,1\}$.
    If $\pi = 0$, $S_0$ calculates $D= {(\llbracket x \rrbracket \cdot {\llbracket y \rrbracket }^{N-1})}^{r_1}\cdot \llbracket r_1+r_2 \rrbracket $.
    If $\pi = 1$, $S_0$ calculates $D= {(\llbracket y \rrbracket \cdot {\llbracket x\rrbracket }^{N-1})}^{r_1}\cdot \llbracket r_2 \rrbracket $. 
    Subsequently, $S_0$ performs a partial decryption of $D$ using \texttt{PDec} to obtain $D_1$, and sends $D$ and $D_1$ to $S_1$.
    \item 
    Upon receiving $D$ and $D_1$, $S_1$ performs a partial decryption of $D$ using \texttt{PDec} to obtain $D_2$.
    Subsequently, $S_1$ obtains $d$ by calling \texttt{TDec} with $D_1$ and $D_2$.
    Next, $S_1$ extracts $\llbracket 0 \rrbracket $ and $\llbracket 1 \rrbracket $ from $tuple_{S_1}$ and refreshes them in $tuple_{S_1}$.
    If $\pi = 0$, $d= r_1\cdot (x-y)+(r_1+r_2)$, otherwise, $d=r_1\cdot (y-x)+r_2$.
    If $d >\frac{N}{2}$, $S_1$ sets $\llbracket {\mu}_0 \rrbracket= \llbracket 0 \rrbracket $.
    Conversely, if $d \le \frac{N}{2}$, $S_1$ sets ${\llbracket \mu}_0 \rrbracket=\llbracket 1 \rrbracket $.
    Finally, $S_1$ sends ${\llbracket \mu}_0 \rrbracket$ to $S_0$.
    \item 
    If $\pi = 0$, $S_0$ sets $\llbracket \mu \rrbracket  = \llbracket {\mu}_0 \rrbracket $.
    Conversely, if $\pi = 1$, $S_0$ extracts $\llbracket 1 \rrbracket $ from $tuple_{S_0}$, and refreshes  it in $tuple_{S_0}$, then sets $\llbracket \mu \rrbracket  = \llbracket 1 \rrbracket \cdot {\llbracket {\mu}_0\rrbracket }^{N-1}$.
\end{enumerate}

\subsection{Secure Sign Bit-Acquisition Protocol (\texttt{SSBA})}

In this subsection, we re-design the \texttt{SSBA} in SOCI by leveraging the proposed FastPaiTD and the offline and online computation mechanism.

Given $\llbracket x \rrbracket $ as input, where $x \in [-2^l,2^l]$, $S_0$ and $S_1$ collaboratively compute ($\llbracket s_{x} \rrbracket ,\llbracket x^{*}\rrbracket ) \leftarrow \texttt{SSBA}(\llbracket x \rrbracket )$ as output. 
$s_{x}$ is the sign bit of $x$, and $x^{*}$ represents the magnitude of $x$.
If $x\ge 0$, $s_{x}=0$ and $x^{*} = x$, otherwise, $s_{x}=1$ and $x^{*} = -x$.
It should be noted that the input is held by $S_0$ and only $S_0$ has the access to the output.
When describing \texttt{SSBA} in Fig. \ref{alg:algorithm-label-SSBA}, we omit the input and output for conciseness.
As shown in Fig. \ref{alg:algorithm-label-SSBA}, \texttt{SSBA} consists of offline phase and online phase.
In the offline phase, $S_0$ constructs a \textit{$tuple_{S_0}$}, which is consist of $r_1$, $r_2$, $\llbracket r_1 
\rrbracket$, $\llbracket r_2 \rrbracket$, $\llbracket -r_1 \cdot r_2 \rrbracket $, $r_3$, $r_4$, $\llbracket r_3+r_4 \rrbracket $, $\llbracket r_4 \rrbracket $, $\llbracket 0 \rrbracket $ and $\llbracket 1 \rrbracket $.
The online phase of \texttt{SSBA} consists of four steps as detailed below.
\begin{enumerate}[(1)]
    \item 
    $S_0$ extracts $\llbracket 0 \rrbracket $ and $\llbracket 1 \rrbracket $ from $tuple_{S_0}$ and refreshes them in $tuple_{S_0}$.
    \item 
    $S_0$ and $S_1$ collaboratively perform $\llbracket s_{x}\rrbracket \leftarrow \texttt{SCMP}(\llbracket x\rrbracket ,\llbracket 0\rrbracket )$. If $x\ge 0$,  $s_{x}=0$, otherwise, $s_{x}=1$.
    \item 
    $S_0$ computes $\llbracket 1-2s_{x}\rrbracket  = \llbracket 1\rrbracket \cdot {\llbracket s_{x}\rrbracket }^{N-2}$.
    \item 
    Finally, $S_0$ and $S_1$ collaboratively perform $\llbracket x^{*}\rrbracket \leftarrow \texttt{SMUL}(\llbracket 1-2s_{x}\rrbracket ,\llbracket x \rrbracket )$. Obviously, $\llbracket x^{*}\rrbracket = (1-2s_{x})\cdot x$. Furthermore, if $x\ge 0$, $x^{*}=x$, otherwise, $x^{*}=-x$.
\end{enumerate}

\subsection{Secure Division Protocol (\texttt{SDIV})}

In this subsection, we re-design the \texttt{SDIV} in SOCI by leveraging the proposed FastPaiTD and the offline and online computation mechanism.

Given $\llbracket x \rrbracket $ and $\llbracket y \rrbracket $ as input, where $x \in [0,2^l]$ and $y \in (0,2^l]$, $S_0$ and $S_1$ collaboratively compute ($\llbracket q \rrbracket ,\llbracket e\rrbracket ) \leftarrow \texttt{SDIV}(\llbracket x \rrbracket ,\llbracket y\rrbracket )$ as output. In the output of \texttt{SDIV}, $q$ represents the quotient of division and $e$ represents the remainder of division, such that $x=q\cdot y + e$. 
It should be noted that the input is held by $S_0$ and only $S_0$ has the access to the output.
When describing \texttt{SDIV} in Fig. \ref{alg:algorithm-label-SDIV}, we omit the input and output for conciseness.
As shown in Fig. \ref{alg:algorithm-label-SDIV}, \texttt{SDIV} consists of offline phase and online phase.
In the offline phase, $S_0$ constructs a \textit{$tuple_{S_0}$}, which is consist of $r_1$, $r_2$, $\llbracket r_1 
\rrbracket$, $\llbracket r_2 \rrbracket$, $\llbracket -r_1 \cdot r_2 \rrbracket $, $r_3$, $r_4$, $\llbracket r_3+r_4 \rrbracket $, $\llbracket r_4 \rrbracket $, $\llbracket 0 \rrbracket $ and $\llbracket 1 \rrbracket $.
The online phase of \texttt{SDIV} consists of seven steps as detailed below.
\begin{enumerate}[(1)]
    \item 
    $S_0$ extracts $\llbracket 0 \rrbracket $ and $\llbracket 1 \rrbracket $ from $tuple_{S_0}$ and
    refreshes them in $tuple_{S_0}$. Subsequently, $S_0$ sets $\llbracket q\rrbracket  = \llbracket 0 \rrbracket $ and $i=l$.
    \item 
    $S_0$ obtains $\llbracket 2^i\cdot y \rrbracket$ by computing $\llbracket c \rrbracket  = {\llbracket y \rrbracket }^{2^i}$, where $i\in \{l, l-1, ..., 1, 0\}$.
    \item 
    $S_0$ and $S_1$ collaboratively perform $\llbracket \mu \rrbracket \leftarrow \texttt{SCMP}(\llbracket x \rrbracket ,\llbracket c \rrbracket )$. If $x \ge 2^i\cdot y$, $\mu = 0$, otherwise, $\mu = 1$.
    \item 
    $S_0$ computes $\llbracket {\mu}' \rrbracket  = \llbracket 1 \rrbracket \cdot {\llbracket \mu \rrbracket }^{N-1}$ and $\llbracket q \rrbracket  = \llbracket q \rrbracket \cdot {\llbracket {\mu}'\rrbracket }^{2^i} $. It is important to note that ${\mu}' = 1-\mu$. Besides, $q=q+ 2^i$ if ${\mu}'=1$, otherwise, $q$ remains unchanged.
    \item 
    $S_0$ and $S_1$ collaboratively perform $\llbracket m \rrbracket \leftarrow \texttt{SMUL}(\llbracket {\mu}' \rrbracket , \llbracket c \rrbracket )$ , where $m = {\mu}'\cdot 2^i\cdot y$. If ${\mu}'=1$ (i.e., $x \ge 2^i\cdot y$), m = $2^i\cdot y$, otherwise, $m=0$.
    \item 
    $S_0$ obtains $\llbracket x \rrbracket= \llbracket x - m \rrbracket$ by computing $\llbracket x \rrbracket  = \llbracket x \rrbracket \cdot{\llbracket m \rrbracket }^{N-1}$.
    If $m=2^i\cdot y$ (i.e., $x \ge 2^i\cdot y$), $x = x-2^i\cdot y$, otherwise, $x$ remains unchanged. 
    Next, sets $i=i-1$.
    Steps (2)-(6) should be repeated until $i<0$.
    \item 
    $S_0$ obtain the remainder $\llbracket e \rrbracket$ by setting $\llbracket e \rrbracket  = \llbracket x \rrbracket $.
\end{enumerate}

\section{Correctness and Security Analysis}\label{Section_6}
\subsection{Correctness Analysis}
In this subsection, we provide the rigorous correctness proofs for the proposed FastPaiTD and the secure outsourced computation protocols in SOCI\textsuperscript{+}.
\begin{theorem}\label{thm_FastPaiTD}
    In FastPaiTD, given two ciphertexts $M_1 \leftarrow \texttt{PDec}(sk_1,c)$ and $M_2 \leftarrow \texttt{PDec}(sk_2,c)$, 
    $\texttt{TDec}(M_1,M_2)$ can correctly recover the plaintext $m$.
\end{theorem}
\begin{proof}
    Before proceeding with the proof, we introduce two important equations.
    The work \cite{ma2021optimized} has proven that their FastPai satisfies the following equation.
    \begin{align} \label{eq1}
        c^{2\alpha}\!\!\!\! \mod N^2 = {(1+N)}^{2\alpha m}.
    \end{align}

    In Eq. (\ref{eq1}), $c$ is a ciphertext, $m$ is the corresponding plaintext, and $\alpha$ is the private key.
    Besides, it is widely recognized that the following equation holds for any integer $m$.
    \begin{align} \label{eq2}
        {(1+N)}^m\!\!\!\! \mod N^2=(1+m\cdot N) \!\!\!\! \mod N^2.
    \end{align}
    
    Now we demonstrate the correctness of the proposed FastPaiTD.
    To simplify the proof process, we adopt the notations $L(x)$ for $\frac{x-1}{N}$ and ${(2\alpha)}^{-1}$ for ${(2\alpha)}^{-1} \!\! \mod N$. 
    The proof process is as follow.

    By substituting $M_1$ and $M_2$ with $c^{sk_{1} } \!\! \mod N^{2}$ and $c^{sk_{2} } \!\! \mod N^{2}$, respectively, we can easily calculate the following equation.
    \begin{align}
        \texttt{TDec}(M_1,M_2) = L(c^{(2\alpha)^{-1} \cdot (2\alpha)}  \cdot c^{\eta \cdot 2\alpha \cdot N}\!\!\! \mod N^2)\!\!\! \mod N.
    \end{align}

    Next, we can obtain the following equation by utilizing Eqs. (\ref{eq1}) and (\ref{eq2}).
    \begin{align}
        \texttt{TDec}(M_1,M_2)=L({(1+N)}^{2\alpha \cdot (2\alpha)^{-1} m }\!\!\! \mod N^2)\!\!\! \mod N.
    \end{align}

    Afterward, we can adopt Eq.(\ref{eq2}) and expand $L(x)$ to get the following equation.
    \begin{align}
        \texttt{TDec}(M_1,M_2)=\frac{(1+mN)-1}{N}\!\!\! \mod N.
    \end{align}

    Finally, we can conclude that $\texttt{TDec}(M_1,M_2) = m \!\! \mod N$.
\end{proof}

\begin{theorem}\label{thm_SMUL}
    Given $\llbracket x \rrbracket $ and $\llbracket y \rrbracket $ as input, where $x,y \in [-2^l,2^l]$, the proposed \texttt{SMUL} protocol correctly outputs $\llbracket x\cdot y\rrbracket $.
\end{theorem}
\begin{proof}
    We assume that the operations in offline phase have been executed correctly. Therefore, our focus now lies on the correctness of online phase.
    
    In step 1, $S_0$ computes $X = \llbracket x \rrbracket \cdot \llbracket r_1 \rrbracket  = \llbracket x+r_1 \rrbracket $ and $Y = \llbracket y\rrbracket \cdot \llbracket r_2 \rrbracket  = \llbracket y+r_2\rrbracket $. Subsequently, $S_0$ computes $C = X^L\cdot Y = \llbracket L\cdot (x+r_1)+y+r_2 \rrbracket $ and $C_1 \leftarrow \texttt{PDec}(sk_1,C)$.

    In step 2, upon receiving $C$ and $C_1$, $S_1$ performs $C_2 \leftarrow \texttt{PDec}(sk_2,C)$. Subsequently, $S_1$ performs $L\cdot (x+r_1)+y+r_2 \leftarrow \texttt{TDec}(C_1,C_2)$, which yields $x+r_1$ and $y+r_2$.

    In step 3, upon receiving $\llbracket (x+r_1) \cdot (y+r_2)\rrbracket $ from $S_1$, $S_0$ computes $\llbracket -r_2x\rrbracket = {\llbracket x\rrbracket}^{-r_2}$ and $\llbracket -r_1y\rrbracket = {\llbracket y\rrbracket}^{-r_1} $. Subsequently, $S_0$ computes $\llbracket (x+r_1)\cdot (y+r_2)\rrbracket  \cdot \llbracket -r_2x\rrbracket  \cdot \llbracket -r_1y\rrbracket  \cdot \llbracket -r_1r_2\rrbracket  = \llbracket (x+r_1)\cdot (y+r_2) - r_2x - r_1y -r_1r_2\rrbracket  = \llbracket x\cdot y\rrbracket $.
\end{proof}

\begin{theorem}\label{thm_SCMP}
    Given $\llbracket x \rrbracket $ and $\llbracket y \rrbracket $ as input, where $x,y \in [-2^l,2^l]$, the proposed \texttt{SCMP} protocol correctly outputs $\llbracket \mu \rrbracket $.
    If $\mu=0$, $x\ge y$, otherwise, $x<y$.
\end{theorem}
\begin{proof}
    According to Fig. \ref{alg:algorithm-label-SCMP}, there are two possible values for D that can be obtained, i.e., $r_1\cdot(x-y+1)+r_2$ and $r_1\cdot(y-x)+r_2$.
    In Fig. \ref{alg:algorithm-label-SCMP}, $r_1$ and $r_2$ are derived from $r_3$ and $r_4$ in $tuple_{S_0}$, hence we have $r_1 \leftarrow {\{0,1\}}^{\sigma}\backslash \{0\}$, $r_2 \le \frac{N}{2} $ and $r_1+r_2>\frac{N}{2}$. 
    Since $r_1 \leftarrow {\{0,1\}}^{\sigma}\backslash \{0\}$, $r_2 \le \frac{N}{2} $, $r_1+r_2>\frac{N}{2}$ and $x,y\in [-2^l,2^l]$, we can easily obtain $0 < r_1\cdot(x-y+1)+r_2 < N$ and $0 < r_1\cdot(y-x)+r_2 < N$.
    
    When $0 < r_1\cdot(x-y+1)+r_2 \le \frac{N}{2}$, it implies that $x-y+1\le0$. Consequently, we have $x<y$, and \texttt{SCMP} outputs $\llbracket 1 \rrbracket $.
    When $\frac{N}{2} < r_1\cdot(x-y+1)+r_2 \le N$, it implies that $x-y+1\ \ge 1$. In this case, we have $x \ge y$, and \texttt{SCMP} outputs $\llbracket 0 \rrbracket $.

    When $0< r_1\cdot(y-x)+r_2 \le \frac{N}{2}$, it implies that $y-x\le 0$. In this case, we have $x\ge y$, and \texttt{SCMP} outputs $\llbracket 0 \rrbracket $.
    When $\frac{N}{2}< r_1\cdot(y-x)+r_2 < N$, it implies that $y-x \ge 1$. In this case, we have $x < y$, and \texttt{SCMP} outputs $\llbracket 1 \rrbracket $.

    Therefore, the proposed \texttt{SCMP} correctly compares $x$ and $y$.
\end{proof}

\begin{theorem}\label{thm_SSBA}
    Given $\llbracket x \rrbracket $ as input, where $x \in [-2^l,2^l]$, the proposed \texttt{SSBA} protocol correctly outputs $\llbracket s_{x}\rrbracket$ and $\llbracket x^{*}\rrbracket$.
    $s_{x}$ is the sign bit of $x$, and $x^{*}$ represents the magnitude of $x$.
    If $x\ge 0$, $s_{x}=0$ and $x^{*} = x$, otherwise, $s_{x}=1$ and $x^{*} = -x$.
\end{theorem}
\begin{proof}
    Given $\llbracket x \rrbracket $ and $\llbracket 0 \rrbracket$ as input, where $x \in [-2^l,2^l]$, according to Theorem \ref{thm_SCMP}, we can easily observe that the \texttt{SCMP}$(\llbracket x \rrbracket, \llbracket 0 \rrbracket)$  outputs $\llbracket 1 \rrbracket$ when $x \in [-2^l,0)$ and outputs $\llbracket 0 \rrbracket$ when $x \in [0,2^l]$.
    Therefore, if $x \in [0,2^l]$, $s_{x} = 0$, otherwise, $s_{x} = 1$.

    According to Theorem \ref{thm_SMUL}, we can correctly obtain $\llbracket (1-2\cdot s_{x})\cdot x\rrbracket \leftarrow \texttt{SMUL}(\llbracket 1-2\cdot s_{x} \rrbracket, \llbracket x \rrbracket)$.
    When $x\ge 0$, $(1-2\cdot s_{x})\cdot x = x$ as $(1-2\cdot s_{x})=1$. When $x< 0$, $(1-2\cdot s_{x})\cdot x = -x$ as $(1-2\cdot s_{x})=-1$.

    Therefore, the proposed \texttt{SSBA} correctly outputs $\llbracket s_{x}\rrbracket$ and $\llbracket x^{*}\rrbracket$.
\end{proof}

\begin{theorem}\label{thm_SDIV}
     Given $\llbracket x \rrbracket $ and $\llbracket y \rrbracket $ as input, where $x \in [0,2^l]$ and $y \in (0,2^l]$, the proposed \texttt{SDIV} protocol correctly outputs $\llbracket q \rrbracket$ and $\llbracket e \rrbracket$.
     $q$ is the quotient of division, and $e$ is the remainder of division, i.e., $x=q\cdot y + e$.
\end{theorem}
\begin{proof}
    It is widely recognized that any quotient $q$ satisfying $q \in [0,2^l]$ and $0 \le x - q\cdot y < y$ can be represented as $ {\textstyle \sum_{0}^{i=l}q_i\cdot 2^i} $, where $q_i \in \{0,1\}$.
    As shown in the loop of Fig. \ref{alg:algorithm-label-SDIV}, for any $i \in \{l,l-1,...,1,0\}$, if $x \ge 2^i \cdot y$, then we have $\mu ' =1$, $q_i = 1\cdot 2^i$ and $x = x - 1\cdot 2^i \cdot y$, otherwise, $\mu ' =0$, $q_i = 0\cdot 2^i$ and $x = x - 0\cdot 2^i \cdot y$.
    We observe that ${\textstyle \sum_{0}^{i=l}q_i\cdot 2^i} = {\textstyle \sum_{0}^{i=l} \mu ' \cdot 2^i}$, where $\mu ' \in \{0,1\}$. Therefore, any $q\in [0,2^l]$ can be represented by $\textstyle \sum_{0}^{i=l} \mu ' \cdot 2^i$.
    Since ${\textstyle \sum_{0}^{i=l}q_i\cdot 2^i} = q$, $e = x - y \cdot  \sum_{0}^{i=l} q_i\cdot 2^i $.
    Therefore, the proposed \texttt{SDIV} protocol correctly outputs $\llbracket q \rrbracket$ and $\llbracket e \rrbracket$.
\end{proof}
\subsection{Security Analysis}
Liu \textit{et al.} \cite{liu2016privacy} has proven the semantic security of their Paillier cryptosystem with partial decryption (PCPD) in POCF. 
In this paper, we adopt the same method used in POCF \cite{liu2016privacy}
to prove the security of FastPaiTD.
In SOCI\textsuperscript{+}, we assume that $S_0$ is not colluding with $S_1$. 
Following the approach of POCF \cite{liu2016privacy}, we define the semantic security model for FastPaiTD. 

\begin{myDef}
Let $\zeta =$ (\texttt{NGen}, \texttt{keygen}, \texttt{Enc}, \texttt{Dec}, \texttt{PDec}, \texttt{TDec}) be a public key cryptosystem that supports partial decryption (\texttt{PDec}) and threshold decryption (\texttt{TDec}). Assuming a polynomial adversary $\mathcal{A}$, if $\mathcal{A}$ has negligible advantage in the challenger-adversary game, then $\zeta$ is semantically secure. The challenger-adversary game is defined as follows.
\begin{itemize}
    \item 
    The challenger obtains the public key $pk$ and private key $sk$ of  $\zeta$ by calling \texttt{keygen}. 
    Subsequently, the challenger splits $sk$ into $sk_1$ and $sk_2$, and sends pk and one of $sk_1$ and $sk_2$ to $\mathcal{A}$. 
    \item 
    $\mathcal{A}$ randomly selects two plaintexts $m_0$ and $m_1$ with equal bit-length, and sends them to the challenger through a secure communication channel.
    \item
    The challenger flips a coin to randomly choose a bit $b \in \{0,1\}$, then adopts \texttt{Enc} to encrypt $m_b$ into ciphertext $c$ and sends $c$ to $\mathcal{A}$.
    \item
    
    The $\mathcal{A}$ outputs a bit $b'$. If $b' = b$, $\mathcal{A}$ succeeds, otherwise, $\mathcal{A}$ fails.
\end{itemize}

The advantage of $\mathcal{A}$ in this game is defined as $ \advantage[\xi]{\kappa} =  \left | \prob{b=b'}-\frac{1}{2}  \right | $, where $\kappa$ is a secure parameter.
\end{myDef}

We now formally adopt the method presented in \cite{liu2016privacy} to prove the semantic security of our novel $(2,2)$-threshold Paillier cryptosystem (FastPaiTD).

\begin{theorem}
    Assuming \text{FastPai} is semantically secure, the \text{FastPaiTD} in \ref{FastPaiTD} is also semantically secure.
\end{theorem}

\begin{proof}
    The semantic security of FastPai has been proven in \cite{ma2021optimized}.
    As same as \cite{liu2016privacy}, 
    we assume a probabilistic polynomial-time adversary $\mathcal{A}$ who breaks the semantic security of FastPaiTD with an advantage at most $\epsilon$.
    We also construct a simulator $\mathcal{S}$ with the same time complexity as $\mathcal{A}$.
    Then, $\mathcal{S}$, $\mathcal{A}$ and the challenger perform the following operations.
    \begin{itemize}
        \item 
        The challenger obtains the public key $pk=(N,h)$ of FastPai.
        \item
        $\mathcal{S}$ randomly chooses $sk_1$, where $sk_1 \in [0, N(N-1)]$.
        \item
        $\mathcal{A}$ receives $pk$ and $sk_1$ from $\mathcal{S}$,
        then randomly chooses two plaintexts $m_0$ and $m_1$ with same bit-length, and sends $m_0$ and $m_1$ to $\mathcal{S}$.
        \item
        After receiving $m_0$ and $m_1$ from $\mathcal{A}$, $\mathcal{S}$ sends them to the challenger of FastPai.
        \item
        The challenger randomly chooses a bit $b$, and encrypts $m_b$ into a ciphertext $c$ by calling \texttt{Enc}. Subsequently, the challenger sends $c$ to 
        $\mathcal{S}$.
        \item
        $\mathcal{S}$ sends $c$ to $\mathcal{A}$.
        \item
        $\mathcal{A}$ finally outputs a bit $b'$ as the guess of $\mathcal{S}$ and sends it to $\mathcal{S}$.
    \end{itemize}

    From the view of $\mathcal{A}$, excepting $sk_1$, the distributions of $pk$ and challenger's ciphertexts are as same as in the real semantic security experiment. 
    According to \ref{FastPaiTD}, the real $sk_1 \in [0,2\alpha N]$. Therefore, given $X \in [0,2\alpha N]$ and  $Y \in [0, N(N-1)]$, we can calculate that $X$ and $Y$ have at most $\frac{2\alpha}{N-1}$ statistical distance. Hence, $\mathcal{S}$ breaks the semantic security of FastPaiTD with advantage at least $\epsilon - \frac{2\alpha}{N-1}$. 
\end{proof}

The offline and online mechanism consists of two parts. The first part involves computing the encryption of random numbers and some constants. The second part involves constructing a pre-computation table to speed up \texttt{Enc}. We now prove that the offline and online mechanism is secure, meaning that it does not reveal any information about the plaintext.

\begin{theorem}
    The offline and online mechanism in SOCI\textsuperscript{+} does not leak any information when performing secure outsourced computations.
\end{theorem}
\begin{proof}
    The security of the first part is proven as follow. 
    Each time $S_0$ and $S_1$ extract a ciphertext $\llbracket m \rrbracket$ from $tuple_{S_0}$ and $tuple_{S_1}$, respectively, they immediately adopt the $\llbracket 0 \rrbracket$ in their tuple to compute $\llbracket m' \rrbracket = 
    \llbracket m \rrbracket \cdot \llbracket 0 \rrbracket$. Subsequently, they replace $\llbracket m \rrbracket$ in their tuple with $\llbracket m' \rrbracket$.
    Although $\llbracket m' \rrbracket$ and $\llbracket m \rrbracket$ are the encrypted values of the same number, $\llbracket m' \rrbracket$ is not identical as $\llbracket m \rrbracket$. Therefore, after refreshing a ciphertext, it appears as if $S_0$ and $S_1$ encrypt a message each time. 
    Consequently, the first part of the offline and online mechanism does not disclose any plaintext information.
    
    In the second part, the pre-computation table is constructed from the public key, allowing anyone to create it. 
    Its sole function is to accelerate the \texttt{Enc}. 
    Therefore, the second part of the offline and online mechanism does not disclose any plaintext information. 
\end{proof}

SOCI adopts the simulation paradigm \cite{micali1987play}, also known as the real/ideal model, to prove the security of its \texttt{SMUL} protocol. 
Therefore, we employ the same method to prove the security of \texttt{SMUL} in SOCI\textsuperscript{+}.
Similar to SOCI, SOCI\textsuperscript{+} assumes that $S_0$ and $S_1$ are semi-honest and non-colluding, which means that $S_0$ and $S_1$ may act as adversaries. We use $\mathcal{A}_{\mathcal{S}_0} $ and $\mathcal{A}_{\mathcal{S}_1} $ to denote $S_0$ and $S_1$ as polynomial-time adversaries, respectively.
We now adopt the same method in SOCI to prove the security of \texttt{SMUL} in SOCI\textsuperscript{+}.

\begin{figure*}[!ht]
  \centering
  \subfloat[Comparison of Enc in Different Schemes]{\includegraphics[width=0.31\textwidth]{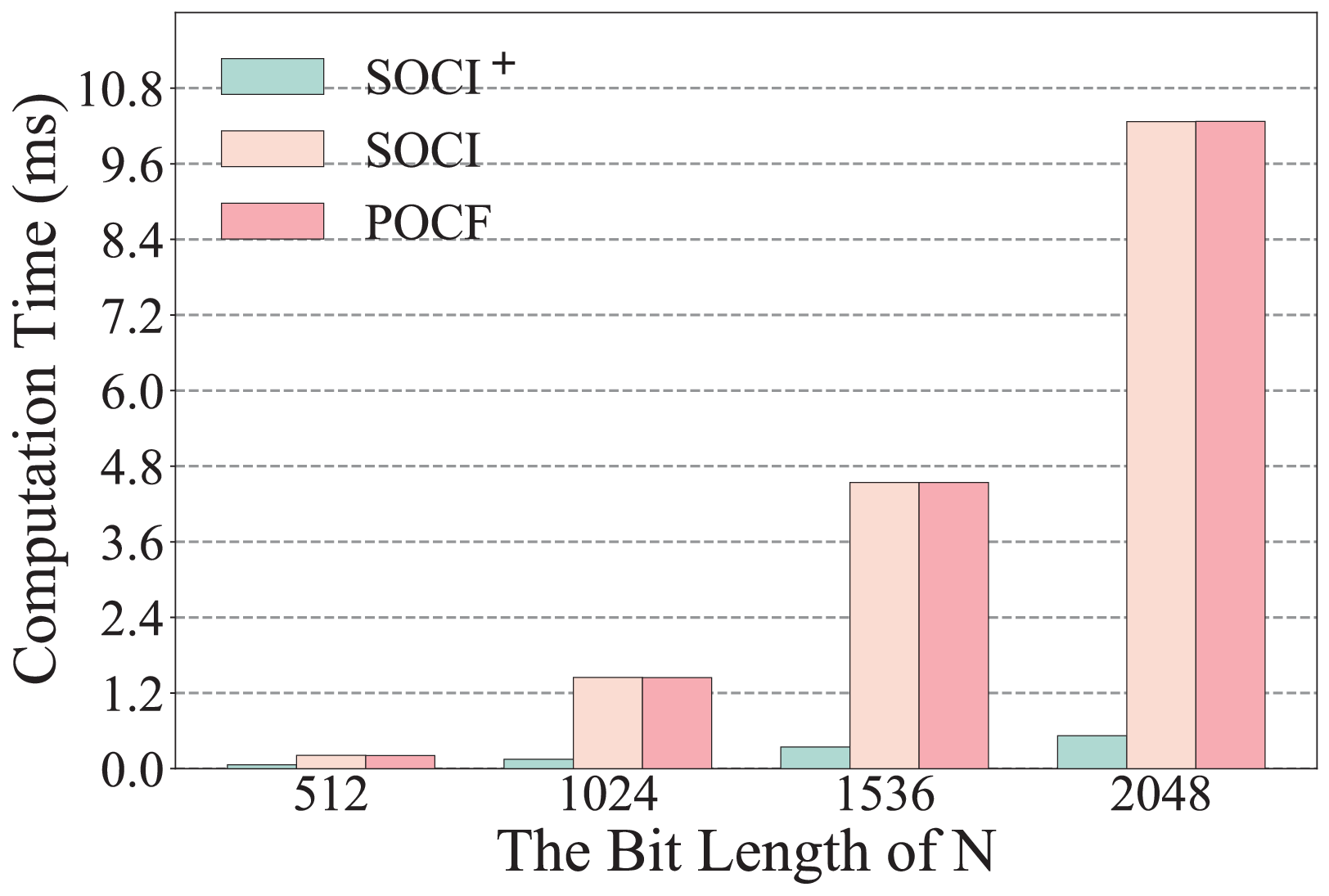}
  \label{fig:Enc}
  }\hfill
  \subfloat[Comparison of Dec in Different Schemes]{\includegraphics[width=0.31\textwidth]{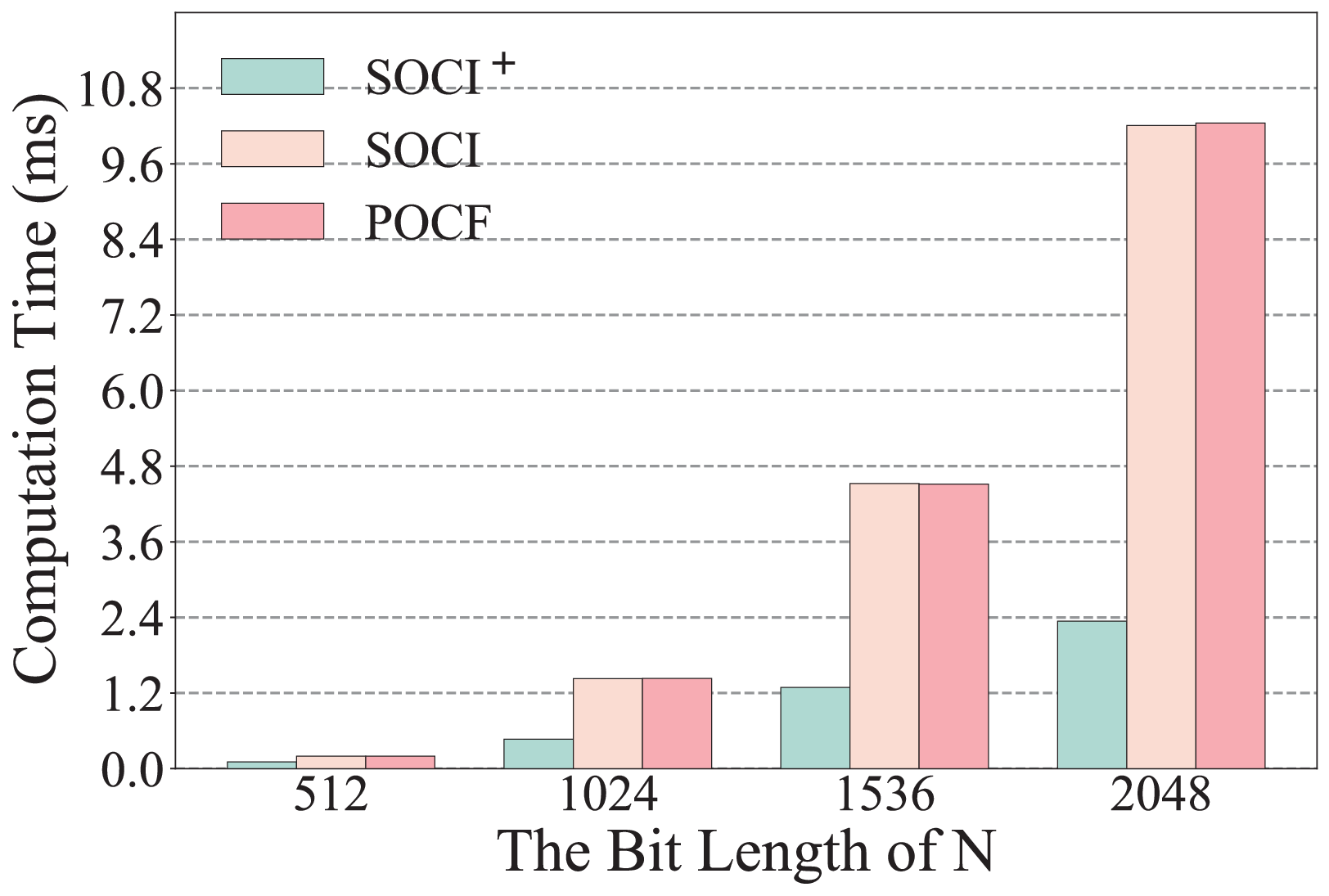}
  \label{fig:Dec}
  }\hfill
  \subfloat[Comparison of PDec (using $sk_2$) in Different Schemes]{\includegraphics[width=0.31\textwidth]{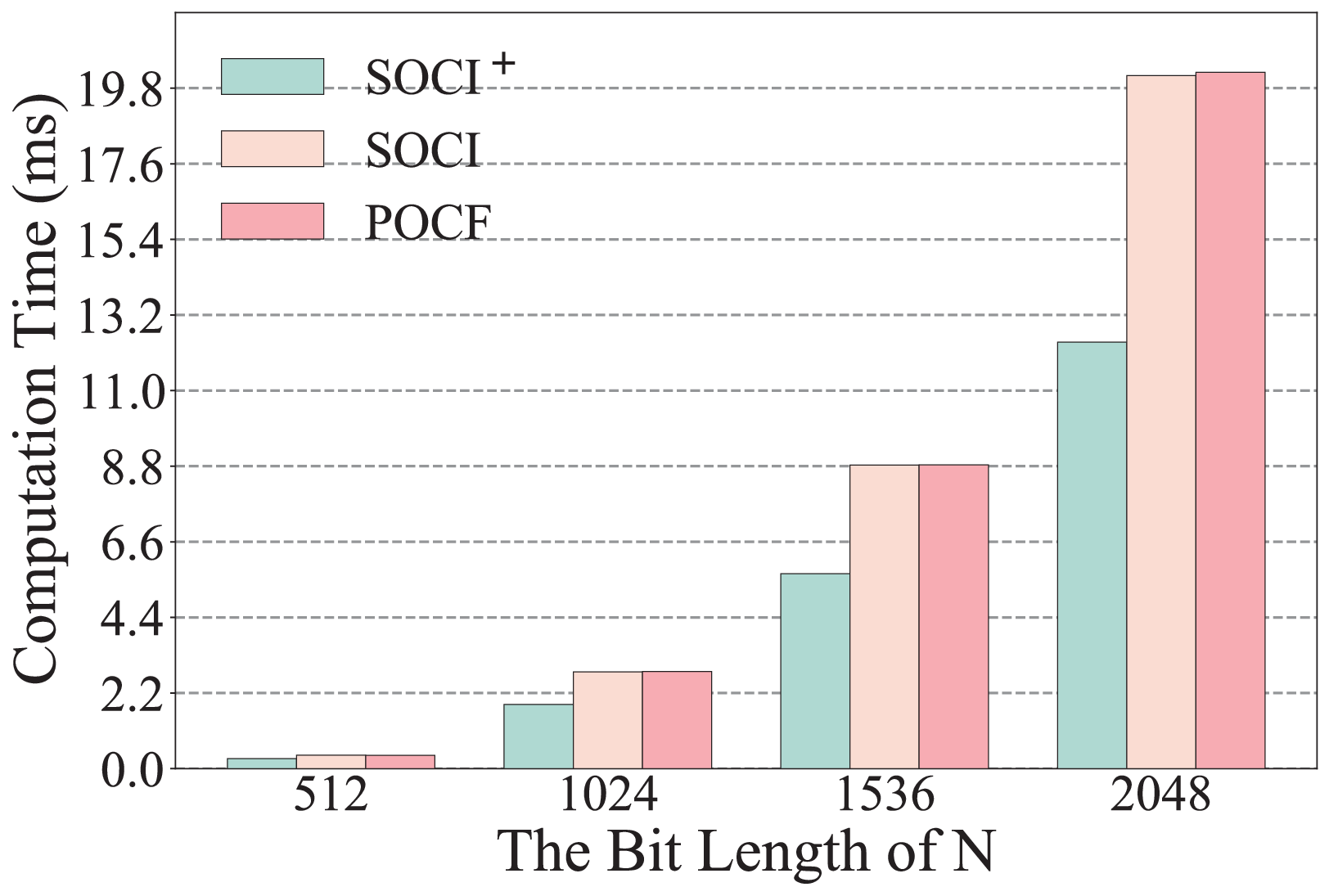}
  \label{fig:PDec}
  }
  \caption{Comparison of Different Threshold Paillier Cryptosystems with a Varying Bit-Length of $N$}
  \label{Comparison of Different}
\end{figure*}

\begin{theorem}
Given ciphertexts $\llbracket x \rrbracket $ and $\llbracket y \rrbracket $, where $x,y\in [-2^l,2^l]$, in the case of semi-honest attackers $\mathcal{A}_{\mathcal{S}_0} $ and $\mathcal{A}_{\mathcal{S}_1} $, the proposed \texttt{SMUL} protocol in SOCI\textsuperscript{+} is able to compute $\llbracket x\cdot y \rrbracket $ securely.
\end{theorem}
\begin{proof}
    To simulate $S_0$ and $S_1$, we construct independent simulators $\mathcal{S}_{\mathcal{S}_0}$ and $\mathcal{S}_{\mathcal{S}_1} $, respectively. 

    $\mathcal{S}_{\mathcal{S}_0}$ simulates the view of $\mathcal{A}_{\mathcal{S}_0}$ as follows.
    \begin{itemize}
        \item 
        $\mathcal{S}_{\mathcal{S}_0}$ takes $\llbracket x \rrbracket $, $\llbracket y \rrbracket $, $\llbracket (x+r_1)\cdot (y+r_2) \rrbracket $ as input, and randomly chooses $\tilde{x}$ and $\tilde{y}$, where $\tilde{x},\tilde{y}\in [-2^l,2^l]$. 
        Besides, $\mathcal{S}_{\mathcal{S}_0}$ also randomly chooses $\tilde{r_1}$, $\tilde{r_2}$, and $\tilde{sk_1}$, where  $\tilde{r_1},\tilde{r_2}, \tilde{sk_1} \leftarrow {\{0,1\}}^{\sigma}$.  Subsequently, $\mathcal{S}_{\mathcal{S}_0}$ randomly chooses $\Tilde{L}$ with ($\sigma$+2) bits. 
        \item
        $\mathcal{S}_{\mathcal{S}_0}$ obtains $\llbracket \tilde{x} \rrbracket $, $\llbracket \tilde{y} \rrbracket $, $\tilde{X}$, $\tilde{Y}$, $\llbracket -\tilde{r_2} \tilde{x} \rrbracket $, $\llbracket -\tilde{r_1} \tilde{y} \rrbracket $ and $\llbracket -\tilde{r_1} \tilde{r_2} \rrbracket $  by calling \texttt{Enc} to encrypt $\tilde{x}$, $\tilde{y}$, $\tilde{x}+\tilde{r_1}$, $\tilde{y}+\tilde{r_2}$, $-\tilde{r_2} \tilde{x}$, $-\tilde{r_1} \tilde{y}$ and $-\tilde{r_1} \tilde{r_2}$, respectively.
        Subsequently, $\mathcal{S}_{\mathcal{S}_0}$ computes $\tilde{C} = {\tilde{X}}^{\tilde{L}} \cdot \tilde{Y}$, and gets $\tilde{C_1} \leftarrow \texttt{PDec}(\tilde{sk_1},\tilde{C})$. 
        \item
        $\mathcal{S}_{\mathcal{S}_0}$ computes $\llbracket \tilde{x} \cdot \tilde{y} \rrbracket  =  \llbracket (x+r_1)\cdot (y+r_2) \rrbracket  \cdot \llbracket -\tilde{r_2} \tilde{x} \rrbracket  \cdot \llbracket -\tilde{r_1} \tilde{y} \rrbracket  \cdot \llbracket -\tilde{r_1} \tilde{r_2} \rrbracket $.
        \item
        Finally, $\mathcal{S}_{\mathcal{S}_0}$ outputs the simulation of $\mathcal{A}_{\mathcal{S}_0}$'s entire view, consisting of $\llbracket \tilde{x} \rrbracket $, $\llbracket \tilde{y} \rrbracket $, $\tilde{X}$, $\tilde{Y}$, $\tilde{C}$, $\tilde{C_1}$, $\llbracket -\tilde{r_2} \tilde{x} \rrbracket $, $\llbracket -\tilde{r_1} \tilde{y} \rrbracket $, $\llbracket -\tilde{r_1} \tilde{r_2} \rrbracket $ and $\llbracket \tilde{x} \cdot \tilde{y} \rrbracket $.
    \end{itemize}
    
We conclude that $\mathcal{S}_{\mathcal{S}_0}$'s view in the ideal world and $\mathcal{A}_{\mathcal{S}_0} $'s view in the real word are computationally indistinguishable since the Paillier cryptosystem in \cite{ma2021optimized} is semantically secure.
 
$\mathcal{S}_{\mathcal{S}_1} $ simulates the view of $\mathcal{A}_{\mathcal{S}_1}$ as follows.
    \begin{itemize}
        \item 
        $\mathcal{S}_{\mathcal{S}_1} $ takes ${\llbracket x+r_1 \rrbracket }^L \cdot {\llbracket y+r_2\rrbracket }$, $({\llbracket x+r_1 \rrbracket }^L \cdot {\llbracket y+r_2\rrbracket })^{sk_1}$ as input, and randomly chooses $\tilde{x}$ and $\tilde{y}$, where $\tilde{x},\tilde{y}\in [-2^l,2^l]$. Subsequently, $\mathcal{S}_{\mathcal{S}_1} $ also randomly chooses $\tilde{r_1}$, $\tilde{r_2}$, $\tilde{sk_1}$ and $\tilde{sk_2}$,
        where  $\tilde{r_1},\tilde{r_2}, \tilde{sk_1}, \tilde{sk_2}  \leftarrow {\{0,1\}}^{\sigma}$. 
        Next, $\mathcal{S}_{\mathcal{S}_1} $ randomly chooses
        $\Tilde{L}$ with ($\sigma$+2) bits. 
        \item
        $\mathcal{S}_{\mathcal{S}_1} $ obtains $\tilde{X}$, $\tilde{Y}$ and $\llbracket (\tilde{x}+\tilde{r_1})\cdot (\tilde{y}+\tilde{r_2})\rrbracket $ by calling \texttt{Enc} to encrypt $\tilde{x}+\tilde{r_1}$, $\tilde{y}+\tilde{r_2}$ and $(\tilde{x}+\tilde{r_1})\cdot (\tilde{y}+\tilde{r_2})$. 
        Moreover, $\mathcal{S}_{\mathcal{S}_1} $ computes $\tilde{C} = {\tilde{X}}^{\tilde{L}} \cdot \tilde{Y}$, $\tilde{C_1} \leftarrow \texttt{PDec}(\tilde{sk_1},\tilde{C})$ and $\tilde{C_2} \leftarrow \texttt{PDec}(\tilde{sk_2},\tilde{C})$.
         \item
        Finally, $\mathcal{S}_{\mathcal{S}_1} $ outputs the simulation of $\mathcal{A}_{\mathcal{S}_1}$'s entire view, consisting of $\tilde{C}$, $\tilde{C_1}$, $\tilde{C_2}$, $\tilde{x}+\tilde{r_1}$, $\tilde{y}+\tilde{r_2}$, $(\tilde{x}+\tilde{r_1}) \cdot (\tilde{y}+\tilde{r_2})$ and $\llbracket (\tilde{x}+\tilde{r_1}) \cdot (\tilde{y}+\tilde{r_2})\rrbracket $.
    \end{itemize}
    
     We conclude that $\mathcal{S}_{\mathcal{S}_1} $'s view in the ideal world 
     and $\mathcal{A}_{\mathcal{S}_1}$'s view in the real word are computationally indistinguishable, since the Paillier cryptosystem is semantically secure and SOCI\cite{zhao2022soci} has proven that the one-time key encryption scheme $x+r$ is able to securely hide $x$.
\end{proof}

\begin{theorem}
    Given ciphertexts $\llbracket x \rrbracket $ and $\llbracket y \rrbracket $, where $x,y\in [-2^l,2^l]$, in the case of semi-honest attackers $\mathcal{A}_{\mathcal{S}_0} $ and $\mathcal{A}_{\mathcal{S}_1} $, the proposed \texttt{SCMP} protocol in SOCI\textsuperscript{+} is able to compare $x$ and $y$ securely.
\end{theorem}
\begin{proof}
    The proposed \texttt{SCMP} in SOCI\textsuperscript{+} roots in the \texttt{SCMP} in SOCI \cite{zhao2022soci}.
    Since SOCI\cite{zhao2022soci} has proven the security of its \texttt{SCMP} and our building blocks (i.e., FastPaiTD and the offline and online mechanism) are secure, the proposed \texttt{SCMP} protocol is able to compare $x$ and $y$ in a secure manner.
\end{proof}

\begin{theorem}
    Given ciphertext $\llbracket x \rrbracket $, where $x\in [-2^l,2^l]$, in the case of semi-honest attackers $\mathcal{A}_{\mathcal{S}_0} $ and $\mathcal{A}_{\mathcal{S}_1} $, the proposed \texttt{SSBA} protocol in SOCI\textsuperscript{+} is able to obtain $\llbracket s_{x}\rrbracket $ and $\llbracket x^{*}\rrbracket$ securely.
\end{theorem}
\begin{proof}
    The proposed \texttt{SSBA} is constructed by calling \texttt{SCMP} and \texttt{SMUL}.
    Since these protocols and the building blocks (i.e., FastPaiTD and the offline and online mechanism) are secure, the proposed \texttt{SSBA} protocol is able to obtain $\llbracket s_{x}\rrbracket $ and $\llbracket x^{*}\rrbracket$ in a secure manner.
\end{proof}

\begin{theorem}
     Given ciphertexts $\llbracket x \rrbracket $ and $\llbracket y \rrbracket $, where $x\in [0,2^l]$ and $y\in (0,2^l]$, in the case of semi-honest attackers $\mathcal{A}_{\mathcal{S}_0} $ and $\mathcal{A}_{\mathcal{S}_1} $, the proposed \texttt{SDIV} protocol in SOCI\textsuperscript{+} is able to obtain $\llbracket q \rrbracket $ and $\llbracket e \rrbracket$ securely.
\end{theorem}
\begin{proof}
    The proposed \texttt{SDIV} is constructed by calling \texttt{SCMP} and \texttt{SMUL}.
    Since these protocols and the building blocks (i.e., FastPaiTD and the offline and online mechanism) are secure, the proposed \texttt{SDIV} protocol is able to obtain $\llbracket q\rrbracket $ and $\llbracket e \rrbracket$ in a secure manner.
\end{proof}

\begin{table*}[!ht]
\setlength\tabcolsep{3.1pt}
\centering
\caption{Comparison of Basically Cryptographic Operations and Storage Costs assuming the Bit-Length of N
is 2048 (112-Bit Security)}
\begin{threeparttable}
\begin{tabular}{@{}cccccccccccccc@{}}
\toprule
Scheme                   &Keygen                & Enc                & Dec & PDec ($sk_1$) \tnote{1} & PDec ($sk_2$) \tnote{1} & TDec   & Addition    &  Scalar-mul\tnote{2} & Subtraction   & PK\tnote{3} & SK\tnote{3}          & Ciphertext\tnote{3} \\ \midrule
SOCI\textsuperscript{+}  & \textbf{148.036} ms  & \textbf{0.522} ms  & \textbf{2.340} ms & \textbf{0.704} ms  & \textbf{12.412} ms & \textbf{0.007} ms & \textbf{0.006} ms & \textbf{0.066} ms  & \textbf{0.058} ms & 0.500 KB           & \textbf{0.055} KB            & \textbf{0.500} KB           \\
SOCI                     &1036.621 ms           & 10.269 ms          & 10.207 ms & 0.705 ms & 20.168 ms & 0.007 ms & 0.006 ms & 0.066 ms   & 0.058 ms     & \textbf{0.250} KB  & 0.250 KB                     & 0.500 KB           \\
POCF                     &1043.951 ms           & 10.273 ms          & 10.246 ms & 0.706 ms & 20.264 ms  & 0.007 ms & 0.006 ms & 0.066 ms   & 0.058 ms     & 0.250 KB           & 0.250 KB                     & 0.500 KB           \\ \bottomrule
\end{tabular}
\begin{tablenotes}
    \footnotesize
    \item[1] PDec ($sk_1$) and PDec ($sk_2$) stand for performing PDec with $sk_1$ and $sk_2$, respectively.
    \item[2] Scalar-mul stands for scalar-multiplication.
    \item[3] PK, SK and Ciphertext stand for size of public key, size of private key and size of ciphertext, respectively.
\end{tablenotes}
\end{threeparttable}
\label{basic_operation_and_storage}
\end{table*}

\section{Experimental Evaluation}\label{Section_7}
SOCI\textsuperscript{+} has protocols similar to the privacy preserving integer calculation protocols in POCF \cite{liu2016privacy}. 
In the rest of this paper, for simplicity, we denote the privacy preserving integer calculation protocols in POCF as POCF.
To evaluate the computation and communication costs, we implement SOCI\textsuperscript{+} (which is open source\footnote{https://github.com/W-Q-Deng/SOCI-plus}), SOCI, and POCF using gmpy2-2.1.0a1 in Python 3.6.8 on two identical servers (CPU: AMD EPYC 7402 24-Core Processor; Memory: 128 GB).
In our experiments, we set $l=32$ and $\sigma = 128$. Specially, we set $l=10$ when evaluating the $\texttt{SDIV}$ protocol.
When constructing a pre-computation table to speed up $\texttt{Enc}$ in SOCI\textsuperscript{+}, we set $b = 5$, and the parameter $len$ is equal to the bit-length of $sk$.
It should be noted that POCF fails to support $\texttt{SSBA}$ and $\texttt{SDIV}$. 
Therefore, we adopt the system architecture of POCF to implement $\texttt{SSBA}$ and $\texttt{SDIV}$ proposed by \cite{liu2016efficient}, and regard them as components of POCF.
We repeat all experiments for 500 times with a single thread and take the average as experimental results.
In the rest of this paper, we adopt $|N|$ to denote the bit length of $N$. When presenting the experimental results in the form of table, we highlight all the best results in bold.

\begin{table*}[!ht]
\centering
\caption{Comparison of Computation Costs and Communication Costs assuming 112-Bit Security}
\begin{tabular}{@{}cccc|ccc|ccc@{}}
\toprule
\multirow{2}{*}{Algorithm} & \multicolumn{3}{c|}{Computation Costs}               & \multicolumn{3}{c|}{Communication Costs}          & \multicolumn{3}{c}{Running Time}                      \\ \cmidrule(l){2-10} 
                           & SOCI\textsuperscript{+} & SOCI        & POCF        & SOCI\textsuperscript{+} & SOCI      & POCF       & SOCI\textsuperscript{+} & SOCI         & POCF         \\ \midrule
\texttt{SMUL}                       & \textbf{15.698} ms      & 84.098 ms   & 92.104 ms   & \textbf{1.498} KB       & 2.498 KB  & 2.498 KB   & \textbf{15.821} ms     & 84.303 ms   & 92.309 ms   \\
\texttt{SCMP}                       & \textbf{19.037} ms      & 52.342 ms   & 53.015 ms   & \textbf{1.498} KB       & 1.499 KB  & 1.499 KB   & \textbf{19.160} ms     & 52.465 ms   & 53.138 ms   \\
\texttt{SSBA}                       & \textbf{34.773} ms      & 157.054 ms  & 155.460 ms  & \textbf{2.997} KB       & 3.996 KB  & 3.997 KB   & \textbf{35.019} ms     & 157.381 ms   & 155.787 ms   \\
\texttt{SDIV}                       & \textbf{382.624} ms     & 1524.189 ms & 8524.647 ms & \textbf{32.965} KB      & 43.959 KB & 244.314 KB & \textbf{385.324} ms    & 1527.790 ms & 8544.661 ms \\ \bottomrule
\end{tabular}
\label{table_protocol_comparison}
\end{table*}

\begin{table}[!ht]
\setlength\tabcolsep{3.1pt}
\centering
\caption{Theoretical Comparison of Communication Costs}
\begin{tabular}{@{}ccccc@{}}
\toprule
Scheme & \texttt{SMUL}      & \texttt{SCMP}      & \texttt{SSBA}      & \texttt{SDIV}                    \\ \midrule
SOCI\textsuperscript{+}  & \textbf{3$|N^2|$} bits & \textbf{3$|N^2|$} bits & \textbf{6$|N^2|$} bits & \textbf{6($l+1$)$|N^2|$} bits        \\
SOCI   & 5$|N^2|$ bits & 3$|N^2|$ bits & 8$|N^2|$ bits & 8($l+1$)$|N^2|$ bits        \\
POCF   & 5$|N^2|$ bits & 3$|N^2|$ bits & 8$|N^2|$ bits & ($3l^2+13l+59$)$|N^2|$ bits \\ \bottomrule
\end{tabular}
\label{table-Theoretical Comparison of Communication Cost}
\end{table}

\subsection{Performance of Different Threshold Paillier Cryptosystem}
In this subsection, we evaluate the performance of different threshold Paillier cryptosystems, which form the foundation of SOCI\textsuperscript{+}, SOCI and POCF. 
    
In our experiments, the size of private key in SOCI\textsuperscript{+} is 448 bits when $|N|=2048$, thus the size of private key in SOCI\textsuperscript{+} is about 0.055 KB.
A smaller private key in SOCI\textsuperscript{+} leads to faster \texttt{PDec}.
Figs. \ref{Comparison of Different}\subref{fig:Enc} and \ref{Comparison of Different}\subref{fig:Dec} compare the computation costs of \texttt{Enc} and \texttt{Dec} with different bit-length of N among SOCI\textsuperscript{+}, SOCI and POCF. The results show that SOCI\textsuperscript{+} has fastest encryption and decryption. 

In SOCI\textsuperscript{+} and SOCI, we can compute $M_1 \leftarrow \texttt{PDec}(c,sk_1)$ and $M_2 \leftarrow \texttt{PDec}(c,sk_2)$ to partially decrypt a ciphertext c, respectively. After obtaining $M_1$ and $M_2$, the corresponding plaintext $m$ can be obtained by computing $\texttt{TDec}(M_1,M_2)$. In \cite{liu2016privacy}, POCF has the operations of $\texttt{PDec1}$ and $\texttt{PDec2}$, where $\texttt{PDec1}$ is equivalent to $\texttt{PDec}(c,sk_1)$, and $\texttt{PDec2}$ integrates $\texttt{PDec}(c,sk_2)$ and $\texttt{TDec}(M_1,M_2)$.
For convenience, when describing POCF, we adopt \texttt{PDec} and \texttt{TDec} instead of \texttt{PDec1} and \texttt{PDec2}. For SOCI\textsuperscript{+}, SOCI, and POCF, $sk_1$ is set to be the same number with $\sigma$ bits. 
For SOCI and POCF, we set $sk_2 = \lambda \cdot ({\lambda}^{-1}\!\!\! \mod  N) - sk_1$, where $\lambda$ is the private key of SOCI and POCF, and we set $sk_2 = ( (2\alpha)^{-1}\!\!\!  \mod N) \cdot (2\alpha) - sk_1$ for SOCI\textsuperscript{+}.
Table \ref{basic_operation_and_storage} presents the computation costs comparison of \texttt{PDec} and \texttt{TDec} among SOCI\textsuperscript{+}, SOCI and POCF. The computation costs of \texttt{PDec} (using $sk_1$) and \texttt{TDec} are almost the same in all schemes, but SOCI\textsuperscript{+} achives best performance in \texttt{PDec} (using $sk_2$). 
Fig. \ref{Comparison of Different}\subref{fig:PDec} presents an intuitive comparison of $\texttt{PDec}$ (using $sk_2$), demonstrating that SOCI\textsuperscript{+} outperforms the other two schemes and improves the computation costs by approximately $1.6$ times compared to SOCI
 when $|N| = 2048$.

\subsection{Evaluations for Secure Outsourced Computation Protocols}
In this subsection, we compare the computation costs, communication costs, and running time of \texttt{SMUL}, \texttt{SCMP}, \texttt{SSBA} and \texttt{SDIV} to evaluate their performance. We define the running time as the sum of computation time and communication time, and assuming a bandwidth with 100 Mbps.

Table \ref{table_protocol_comparison} presents the comparison of computation costs, communication costs, and running time among SOCI\textsuperscript{+}, SOCI and POCF. 
The experimental results demonstrate that SOCI\textsuperscript{+} outperforms the other two schemes. %Specifically, computation costs of SMUL, SCMP, SSBA and SDIV in SOCI+ are about 18.7\%, 36.4\%, 22.1\% and 25.1\% of SOCI’s.
Specifically, experimental results indicate that SOCI\textsuperscript{+} improves the computation costs by $2.7 - 5.4$ times compared to SOCI. 
Figs. \ref{Computation Cost of Different Schemes}\subref{fig:local_SMUL}, \ref{Computation Cost of Different Schemes}\subref{fig:local_SCMP}, \ref{Computation Cost of Different Schemes}\subref{fig:local_SSBA} and \ref{Computation Cost of Different Schemes}\subref{fig:local_SDIV} present the computation costs comparison of \texttt{SMUL}, \texttt{SCMP}, \texttt{SSBA} and \texttt{SDIV}, respectively.
The results indicate that SOCI\textsuperscript{+} outperforms both SOCI and POCF in terms of computation costs, and the advantage of SOCI\textsuperscript{+} increase with $|N|$.

Table \ref{table_protocol_comparison} demonstrates that SOCI\textsuperscript{+} generally reduces communication costs by approximately $25\% - 40\%$ compared to SOCI, except for \texttt{SCMP}.
The experimental results for communication costs of \texttt{SMUL}, \texttt{SCMP}, \texttt{SSBA} and \texttt{SDIV} are presented at Figs. \ref{Communication Cost of Different Schemes}\subref{fig:cost_SMUL}, \ref{Communication Cost of Different Schemes}\subref{fig:cost_SCMP}, \ref{Communication Cost of Different Schemes}\subref{fig:cost_SSBA} and \ref{Communication Cost of Different Schemes}\subref{fig:cost_SDIV}, respectively. 
While the three shcemes has almost the same communication costs for \texttt{SCMP}, SOCI\textsuperscript{+} exhibits significant advantage in other protocols when $N$ is large.
To better understand the differences in communication costs among the three schemes, we present a theoretical analysis of communication costs for SOCI\textsuperscript{+}, SOCI and POCF in Table \ref{table-Theoretical Comparison of Communication Cost}. 
The experimental results align with theoretical analysis of communication costs.

The running time of the proposed protocols is affected by computation power and bandwidth.
As previously mentioned, the running time is the sum of computation time and communication time, and we assume the bandwidth is 100 Mbps.
The experimental results for running time are presented in the right-hand side of Table \ref{table_protocol_comparison}.
The results indicate that SOCI\textsuperscript{+} improves $2.7 - 5.3$ times in terms of running time compared to SOCI.
Figs. \ref{Runing Time Cost of Different Schemes}\subref{fig:running_SMUL}, \ref{Runing Time Cost of Different Schemes}\subref{fig:running_SCMP}, \ref{Runing Time Cost of Different Schemes}\subref{fig:running_SSBA} and \ref{Runing Time Cost of Different Schemes}\subref{fig:running_SDIV} present the experimental results for running time with a varying $|N|$.
The results show that SOCI\textsuperscript{+} outperforms both SOCI and POCF in terms of running time with different $|N|$.

\begin{figure*}[!ht]
  \centering
  \subfloat[Computation Costs of SMUL in Different Schemes]{
  \includegraphics[width=0.23\textwidth]{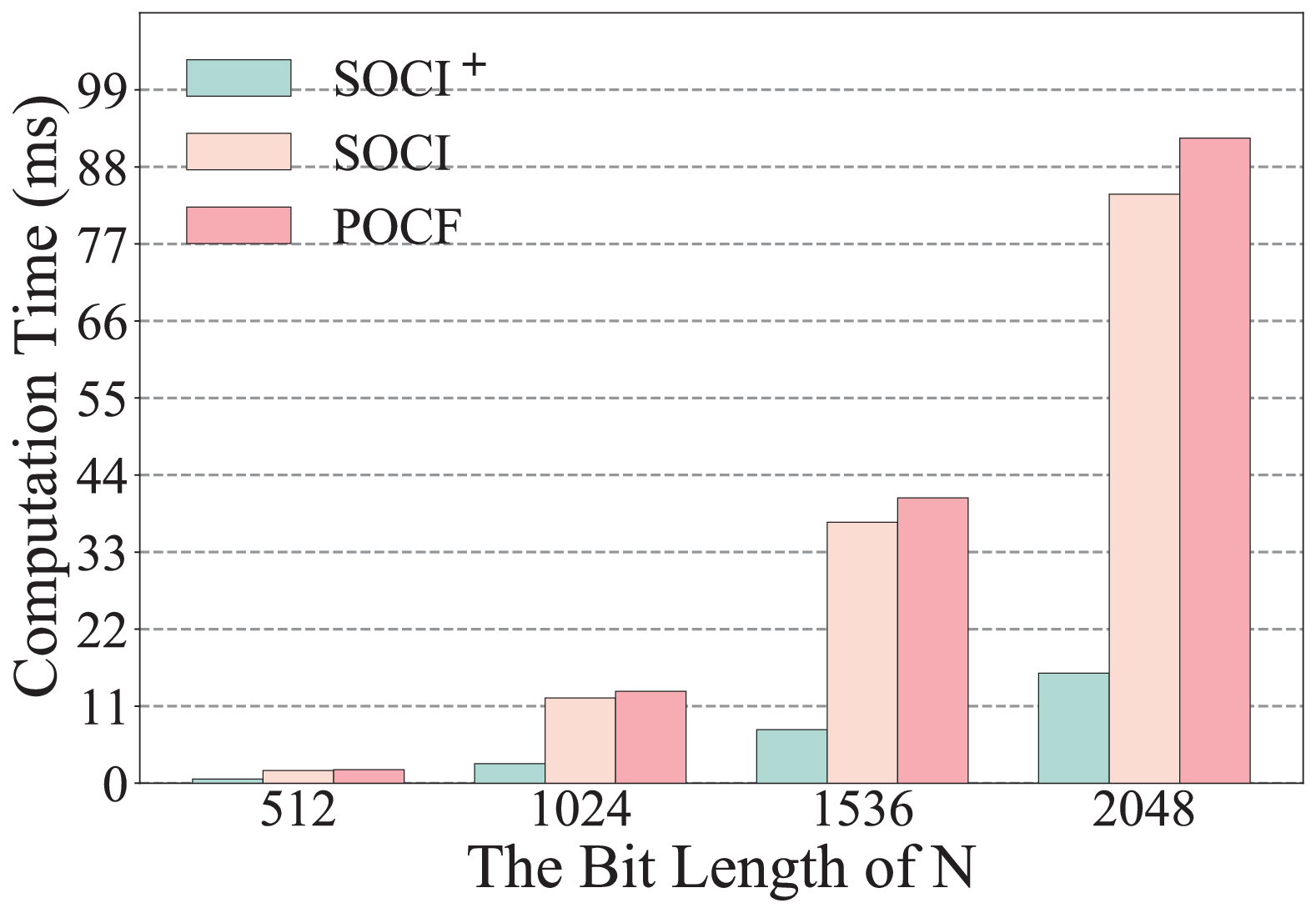}
  \label{fig:local_SMUL}
  }\hfill
  \subfloat[Computation Costs of SCMP in Different Schemes]{
  \includegraphics[width=0.23\textwidth]{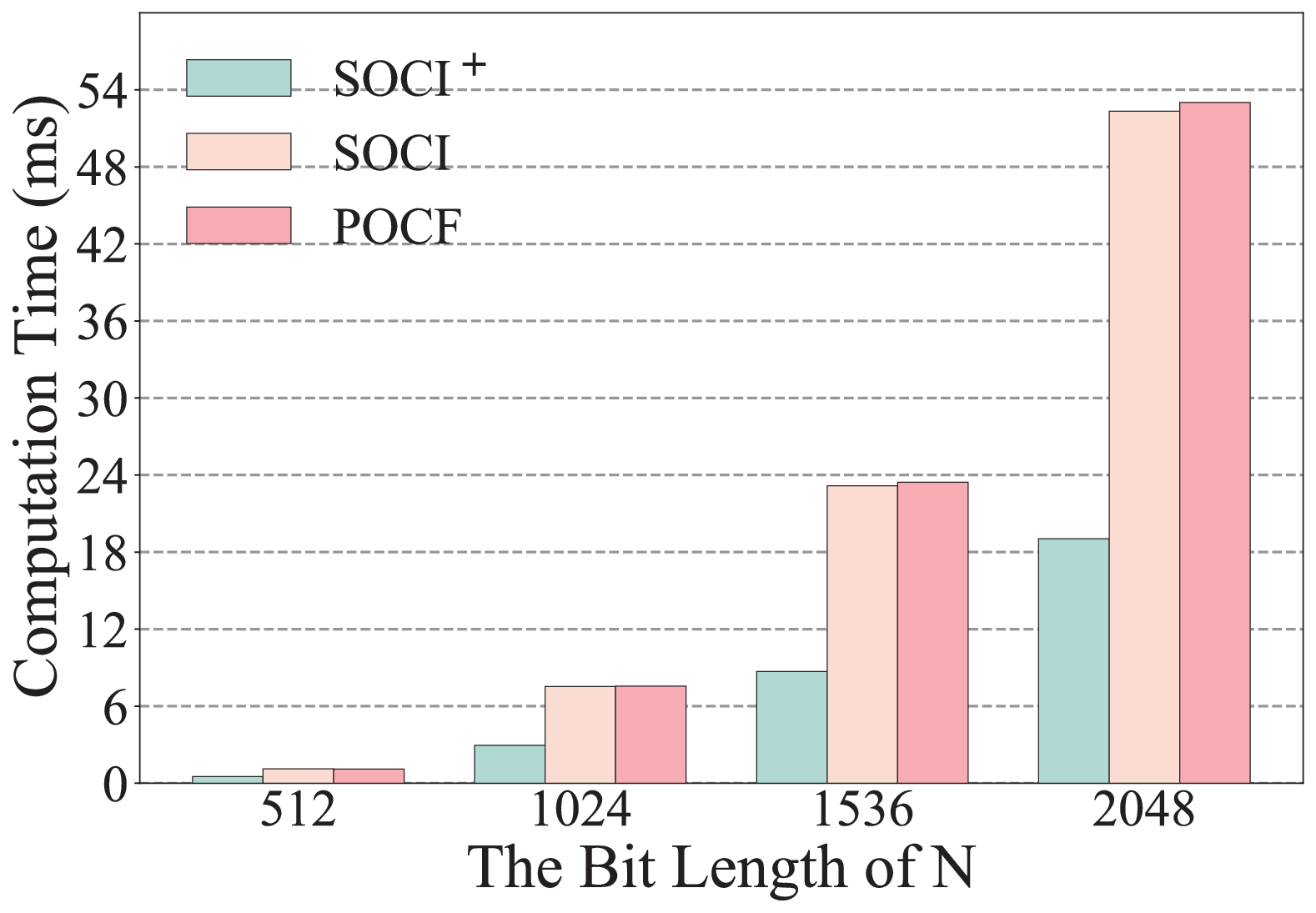}
  \label{fig:local_SCMP}
  }\hfill
  \subfloat[Computation Costs of SSBA in Different Schemes]{
  \includegraphics[width=0.23\textwidth]{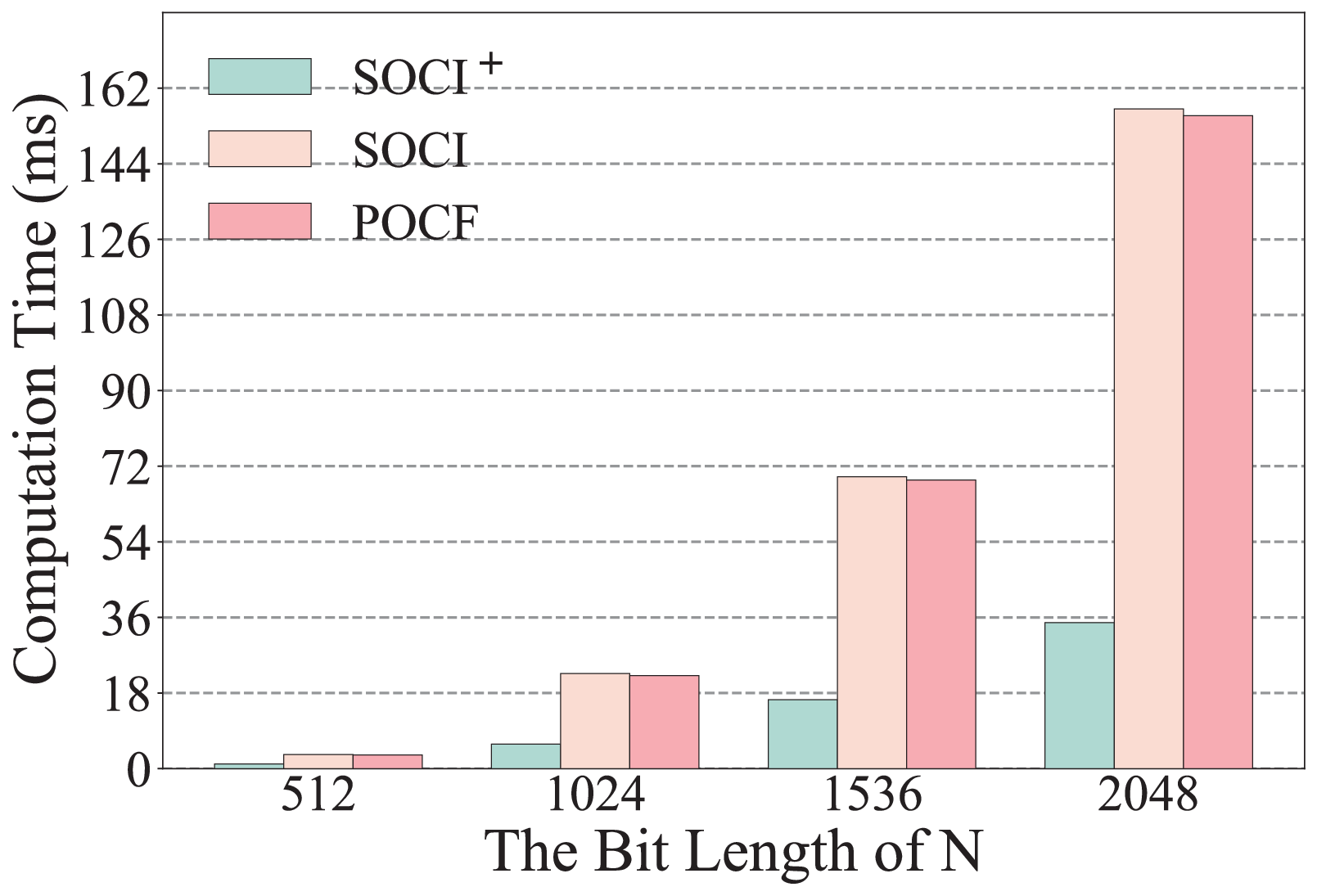}
  \label{fig:local_SSBA}
  }\hfill
  \subfloat[Computation Costs of SDIV in Different Schemes]{
  \includegraphics[width=0.23\textwidth]{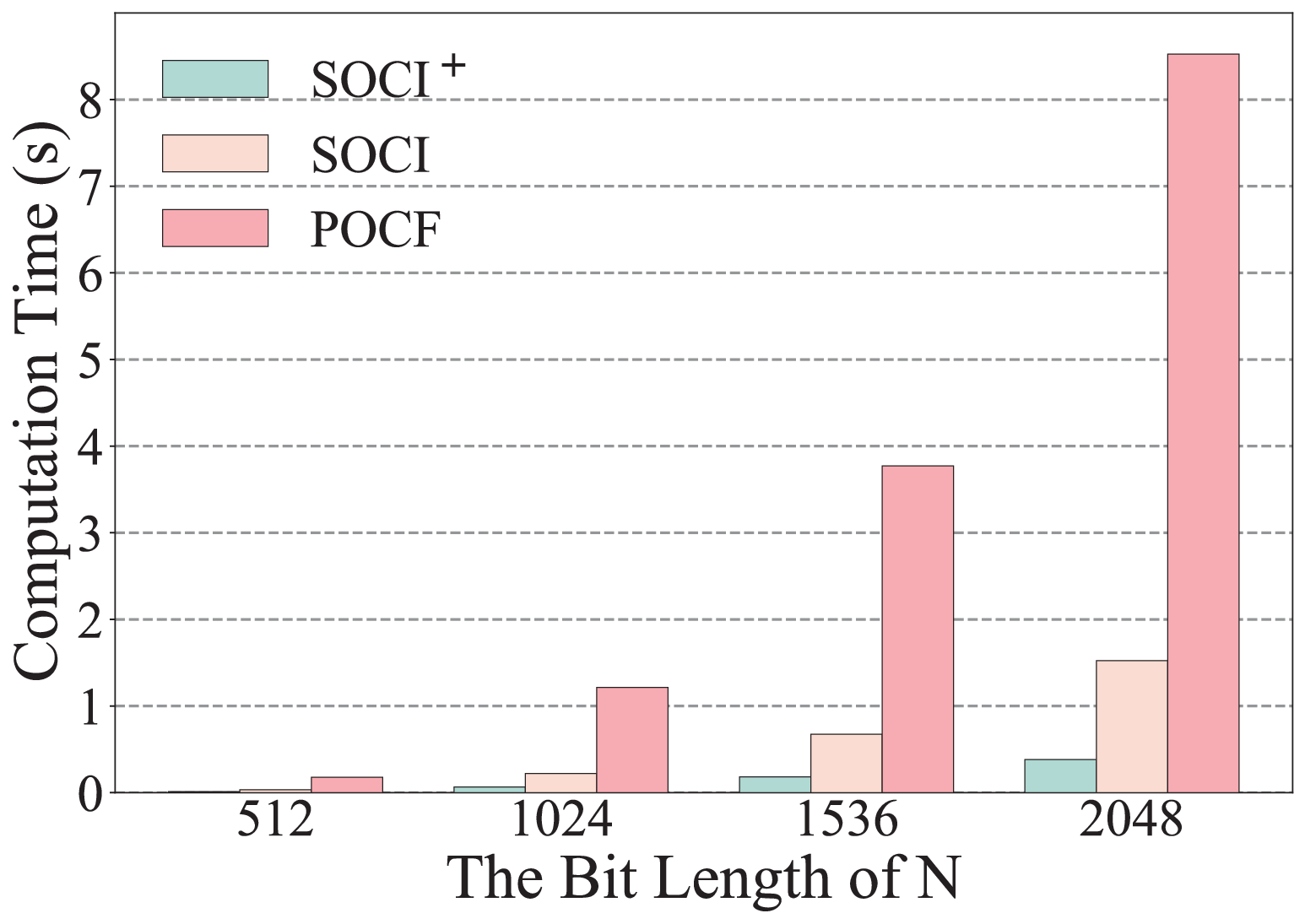}
  \label{fig:local_SDIV}
  }
  \caption{Computation Costs Comparison of Different Schemes with a Varying Bit-Length of $N$}
  \label{Computation Cost of Different Schemes}
\end{figure*}

\begin{figure*}[!ht]
  \centering
  \subfloat[Communication Costs of SMUL in Different Schemes]{
  \includegraphics[width=0.23\textwidth]{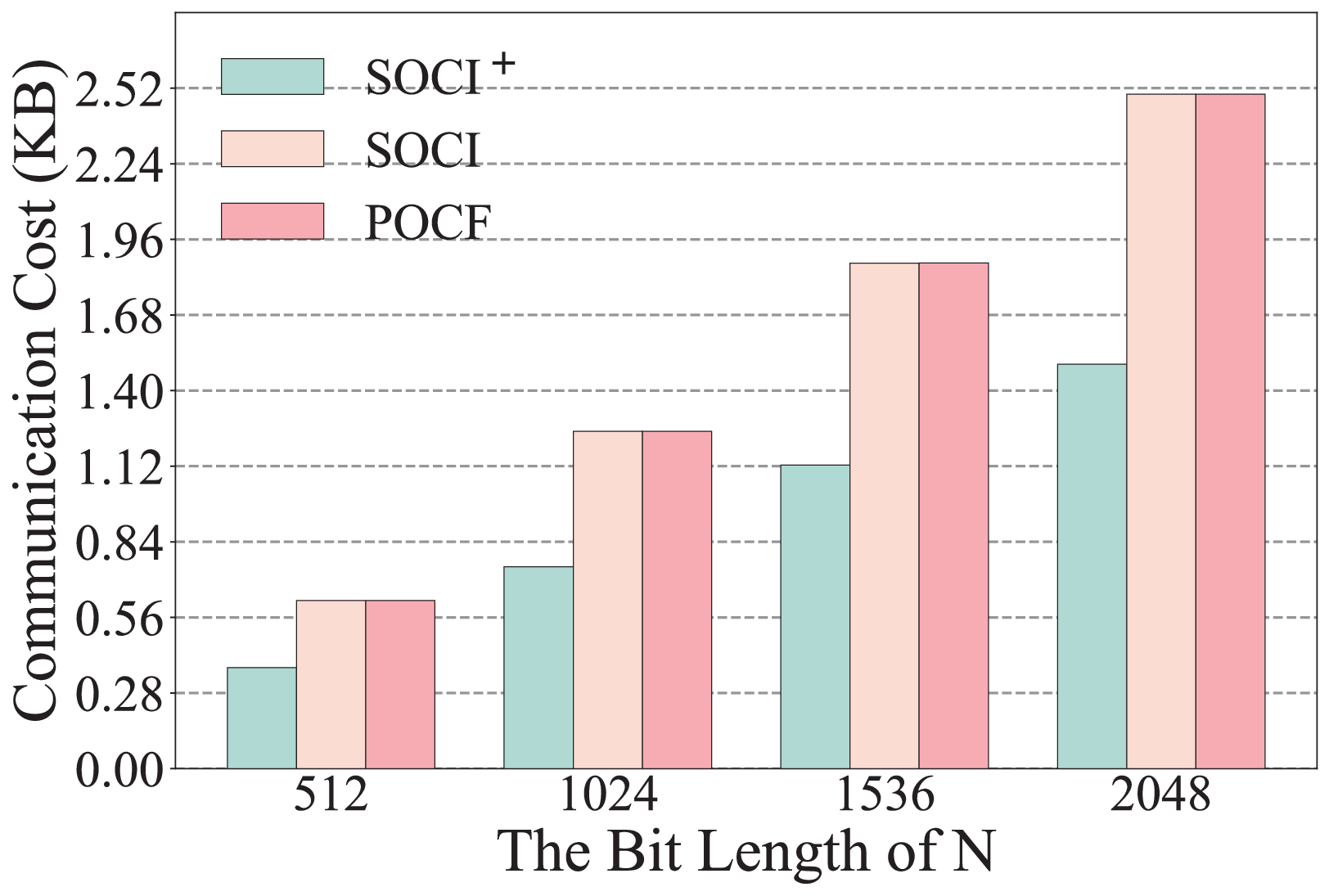}
  \label{fig:cost_SMUL}
  }\hfill
  \subfloat[Communication Costs of SCMP in Different Schemes]{
  \includegraphics[width=0.23\textwidth]{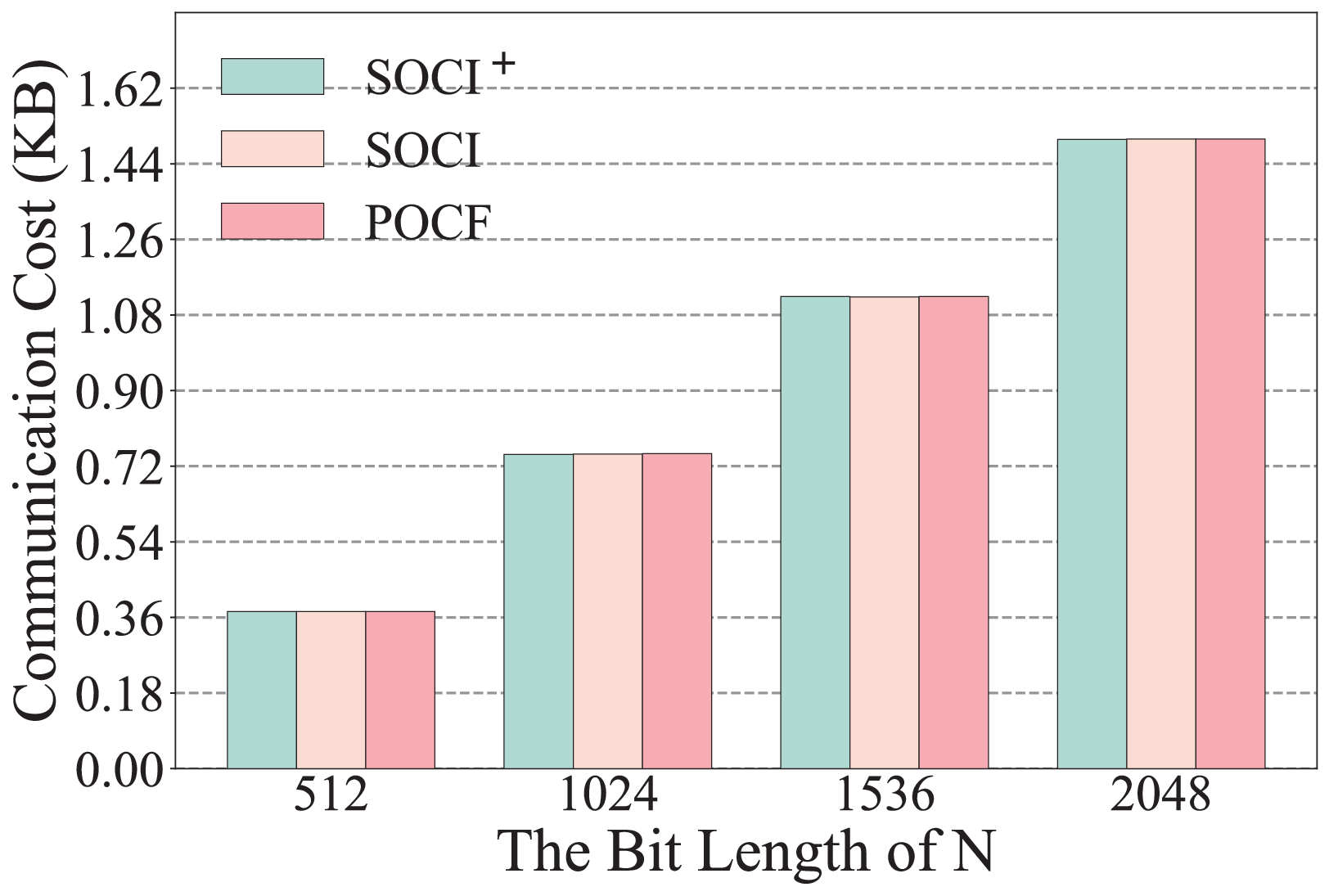}
  \label{fig:cost_SCMP}
  }\hfill
  \subfloat[Communication Costs of SSBA in Different Schemes]{
  \includegraphics[width=0.23\textwidth]{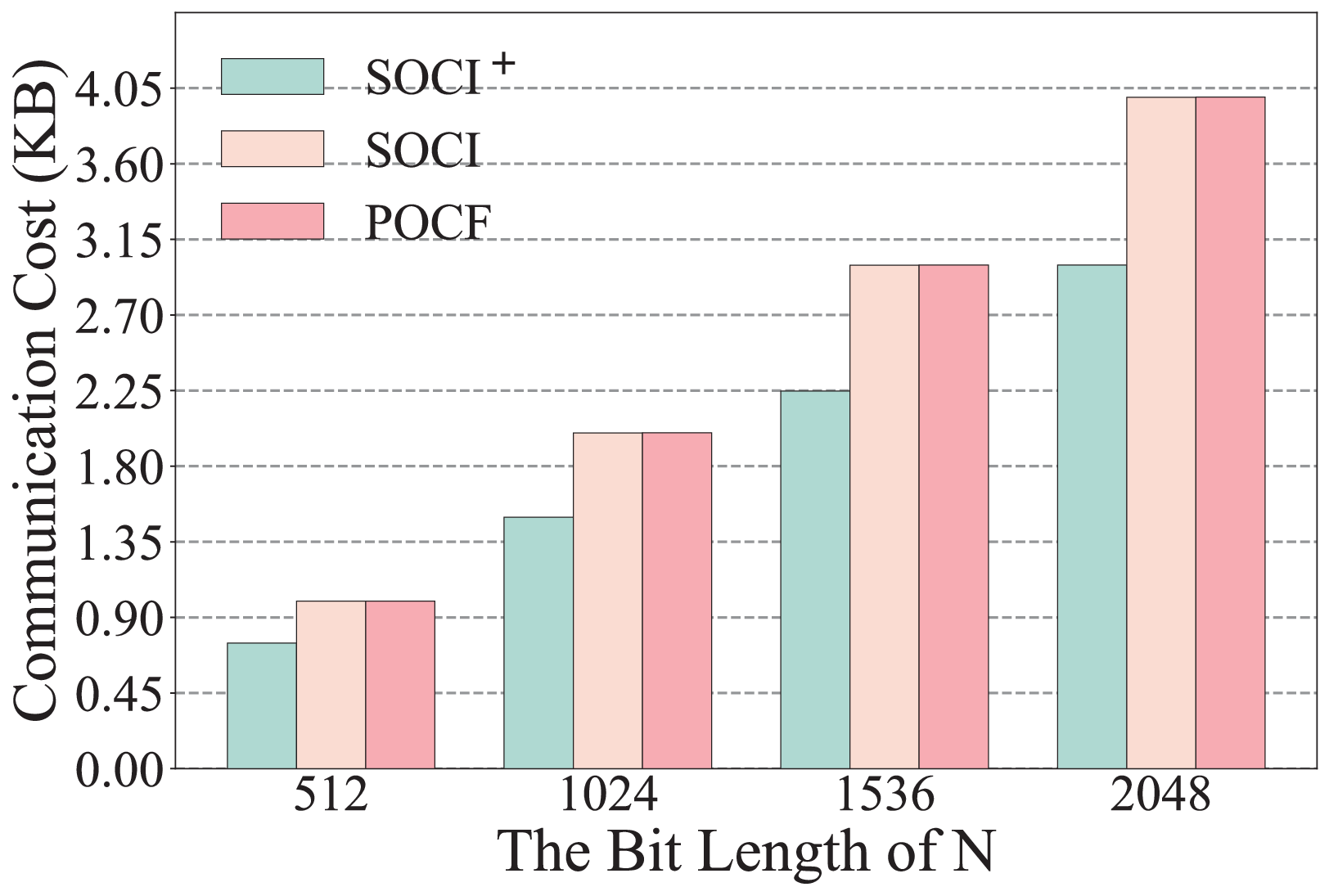}
  \label{fig:cost_SSBA}
  }\hfill
  \subfloat[Communication Costs of SDIV in Different Schemes]{
  \includegraphics[width=0.23\textwidth]{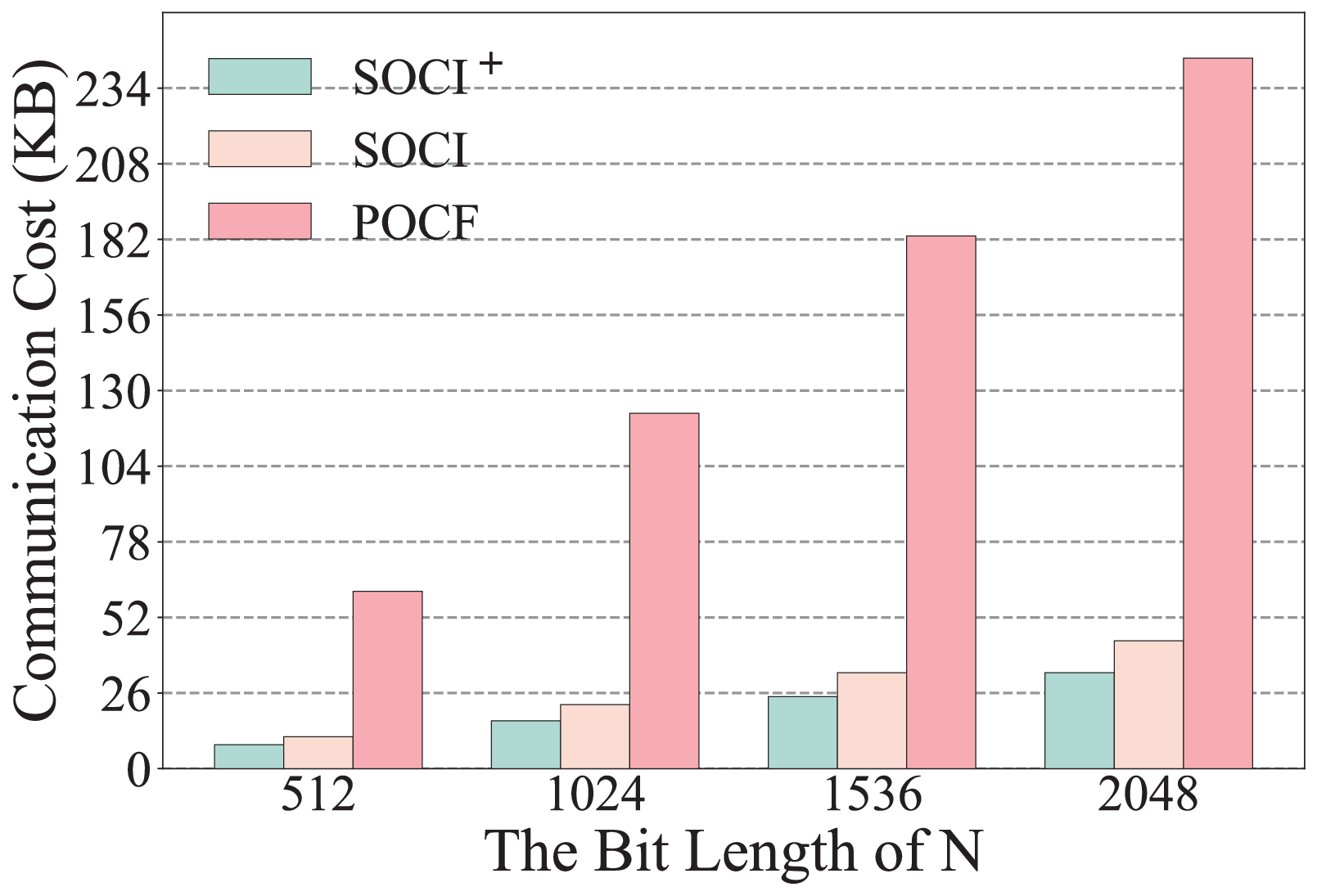}
  \label{fig:cost_SDIV}
  }
  \caption{Communication Costs Comparison of Different Schemes with a Varying Bit-Length of $N$}
  \label{Communication Cost of Different Schemes}
\end{figure*}

\begin{figure*}[!ht]
  \centering
  \subfloat[Running Time of SMUL in Different Schemes]{
  \includegraphics[width=0.23\textwidth]{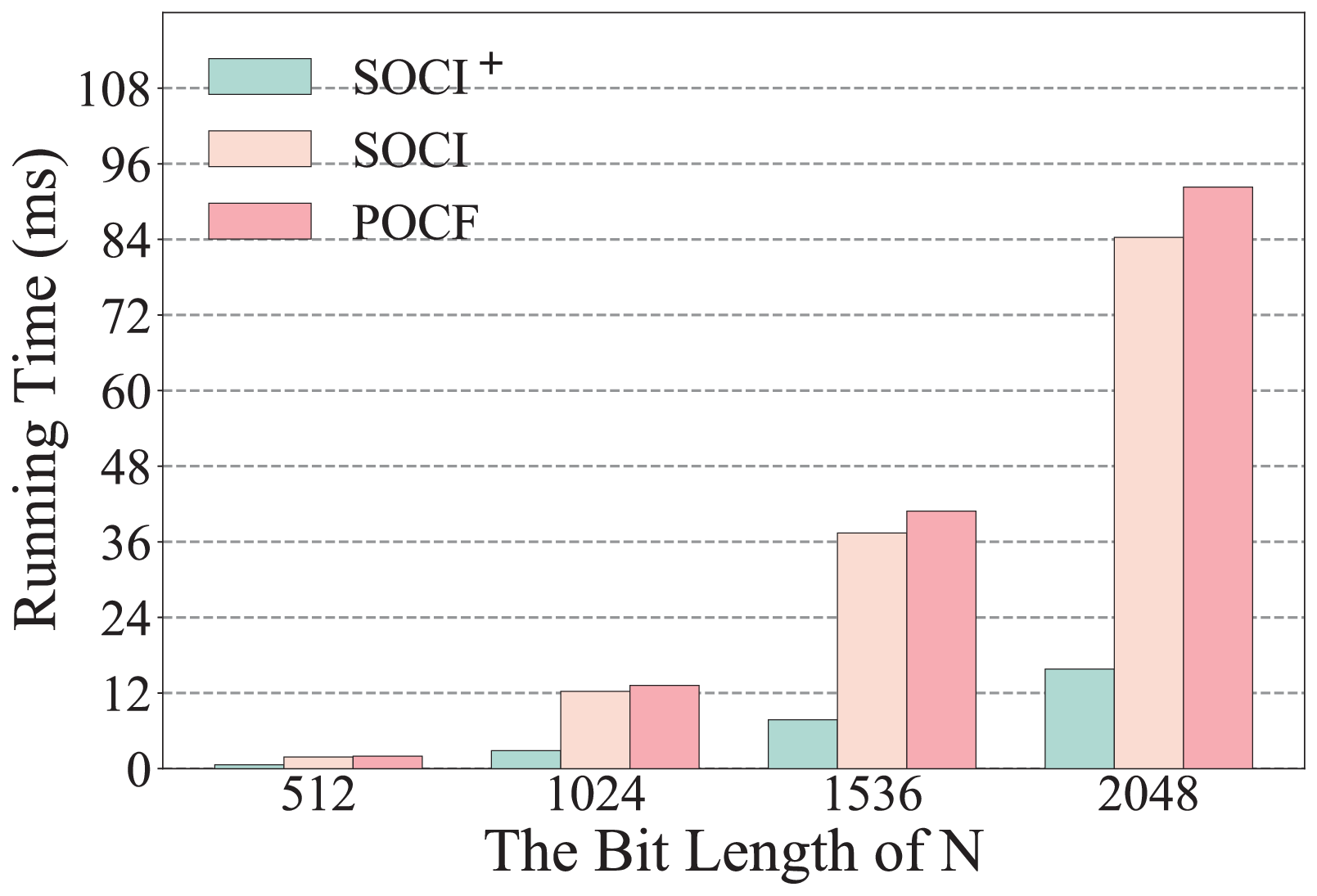}
  \label{fig:running_SMUL}
  }\hfill
  \subfloat[Running Time of SCMP in Different Schemes]{
  \includegraphics[width=0.23\textwidth]{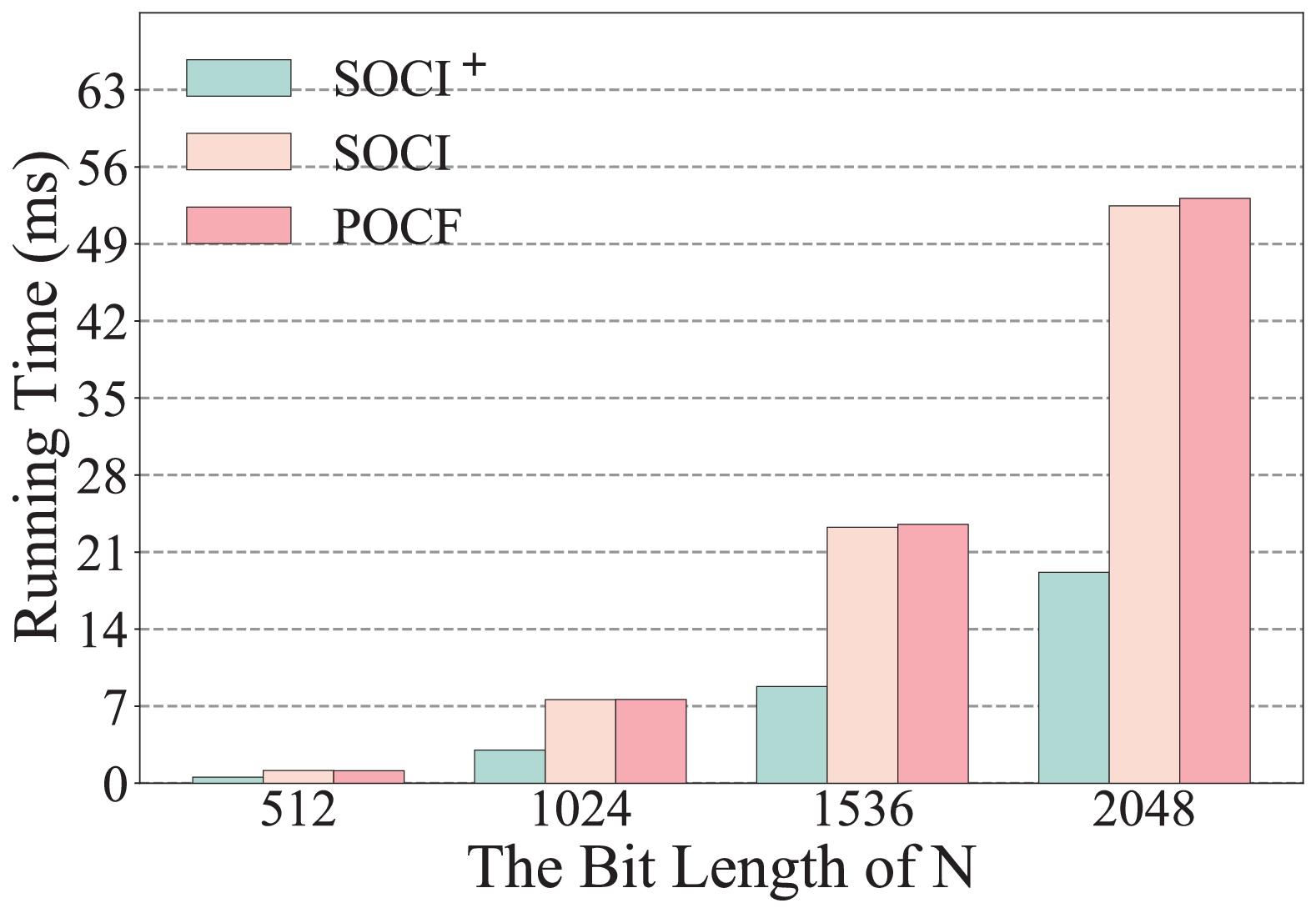}
  \label{fig:running_SCMP}
  }\hfill
  \subfloat[Running Time of SSBA in Different Schemes]{
  \includegraphics[width=0.23\textwidth]{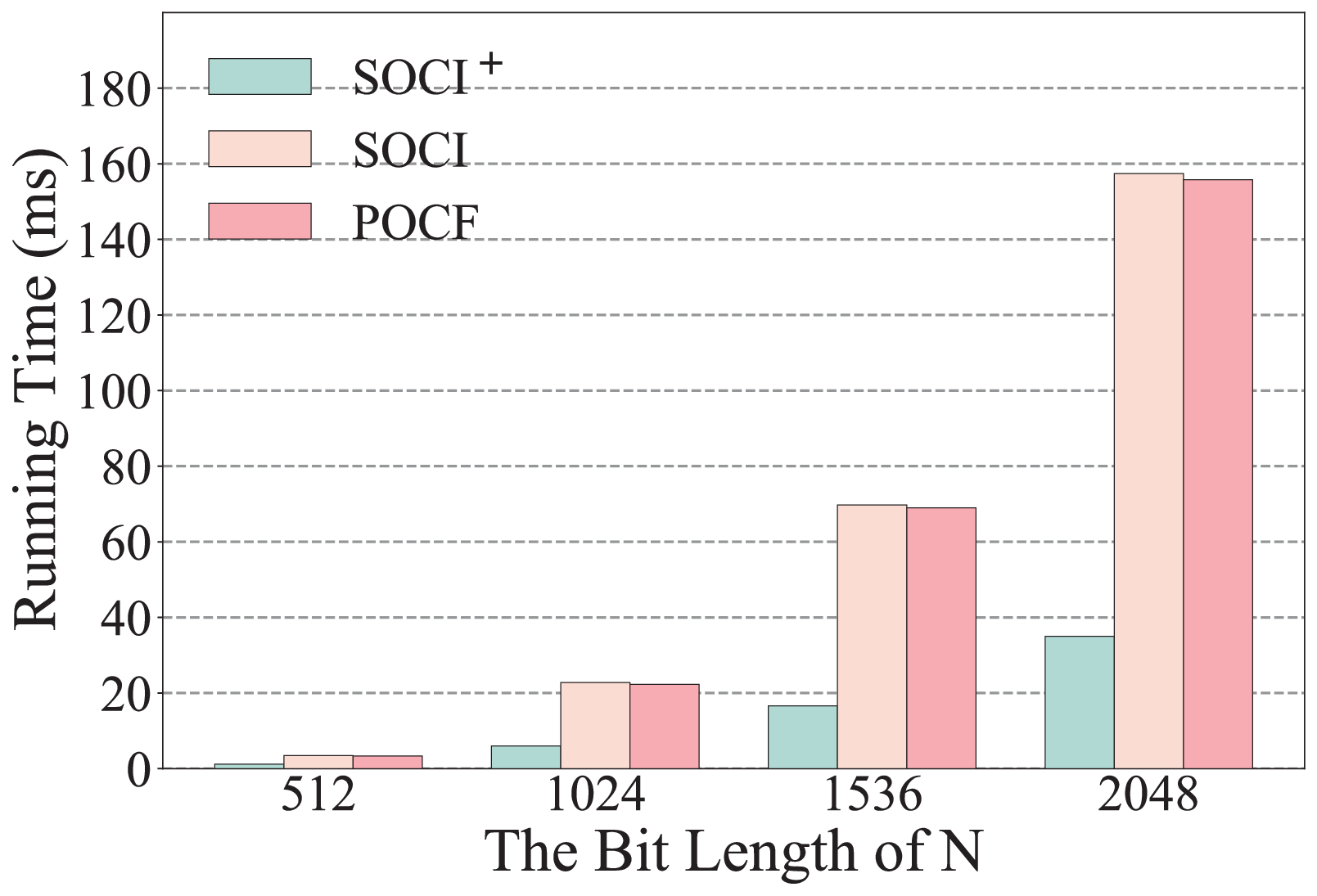}
  \label{fig:running_SSBA}
  }\hfill
  \subfloat[Running Time of SDIV in Different Schemes]{
  \includegraphics[width=0.23\textwidth]{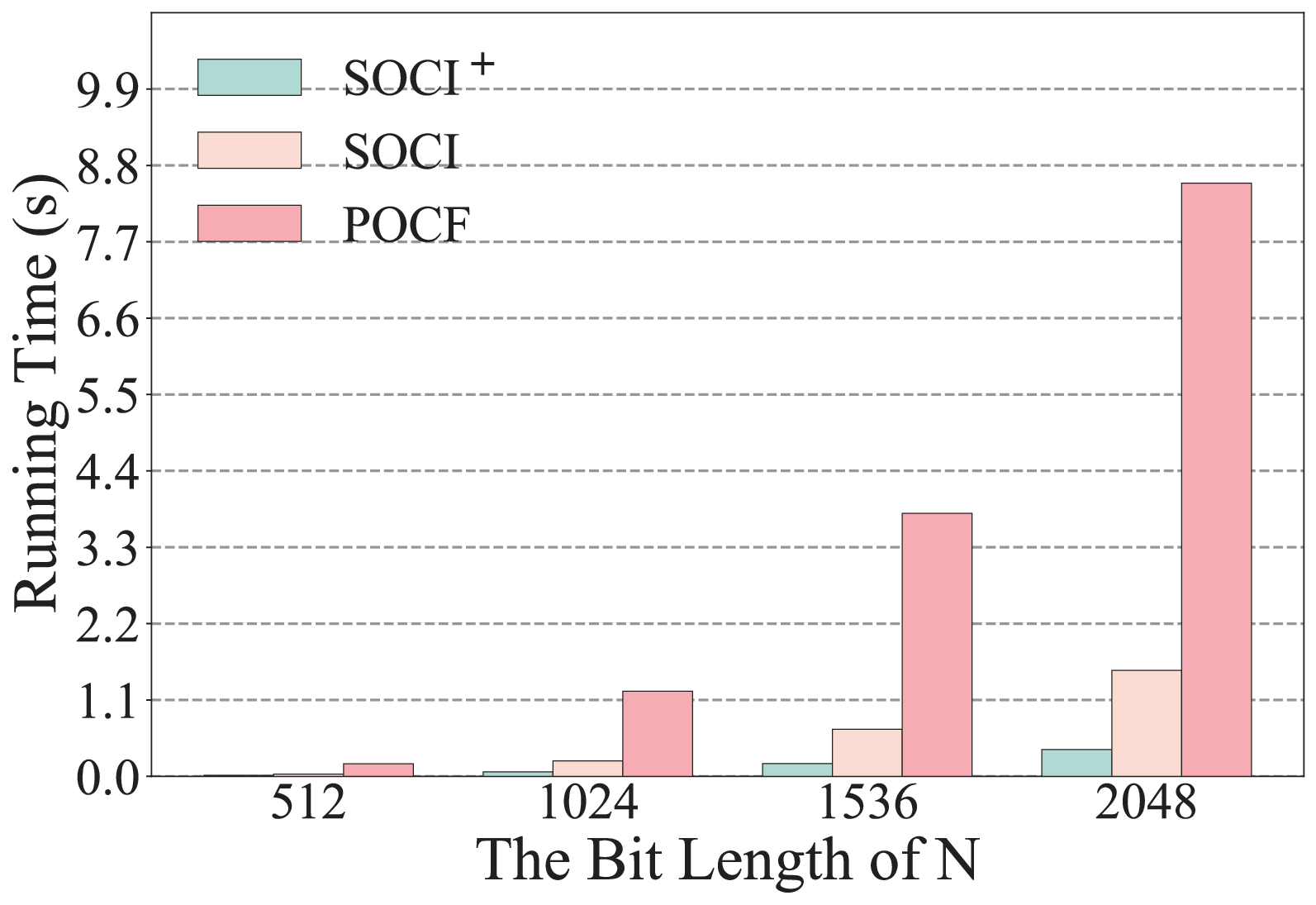}
  \label{fig:running_SDIV}
  }
  \caption{Running Time Comparison of Different Schemes with a Varying Bit-Length of $N$, assuming the Bandwidth is 100 Mbps}
  \label{Runing Time Cost of Different Schemes}
\end{figure*}

\section{Conclusion}\label{Section_8}
In this paper, we proposed SOCI\textsuperscript{+}, an enhanced toolkit for secure outsourced computation on integers. 
Specifically, we designed a novel $(2,2)$-threshold Paillier cryptosystem (FastPaiTD) falling in the twin-server architecture based on the scheme of Ma \textit{et al.} \cite{ma2021optimized} (FastPai).
Additionally, we proposed an offline and online mechanism for SOCI\textsuperscript{+}. Our FastPaiTD and offline and online mechanism significantly improve the performance of secure outsourced computation protocols. 
SOCI\textsuperscript{+} strictly outperforms the state-of-the-art in terms of computation costs and communication costs and is correct and secure. 
In the future work, we will upgrade SOCI\textsuperscript{+} to support floating point arithmetic and more types of secure outsourced computations.

% \section*{Acknowledgments}
% The authors thank the editor in chief, associate editor and anonymous reviewers for their constructive comments and suggestions.

\bibliographystyle{IEEEtran}
\normalem
\bibliography{soci_plus}

% \begin{IEEEbiography}[{\includegraphics[width=1in,height=1.25in,clip,keepaspectratio]{Author.jpg}}]{AuthorName}
% aa
% \end{IEEEbiography}

\begin{IEEEbiography}[{\includegraphics[width=1in,height=1.25in,clip,keepaspectratio]{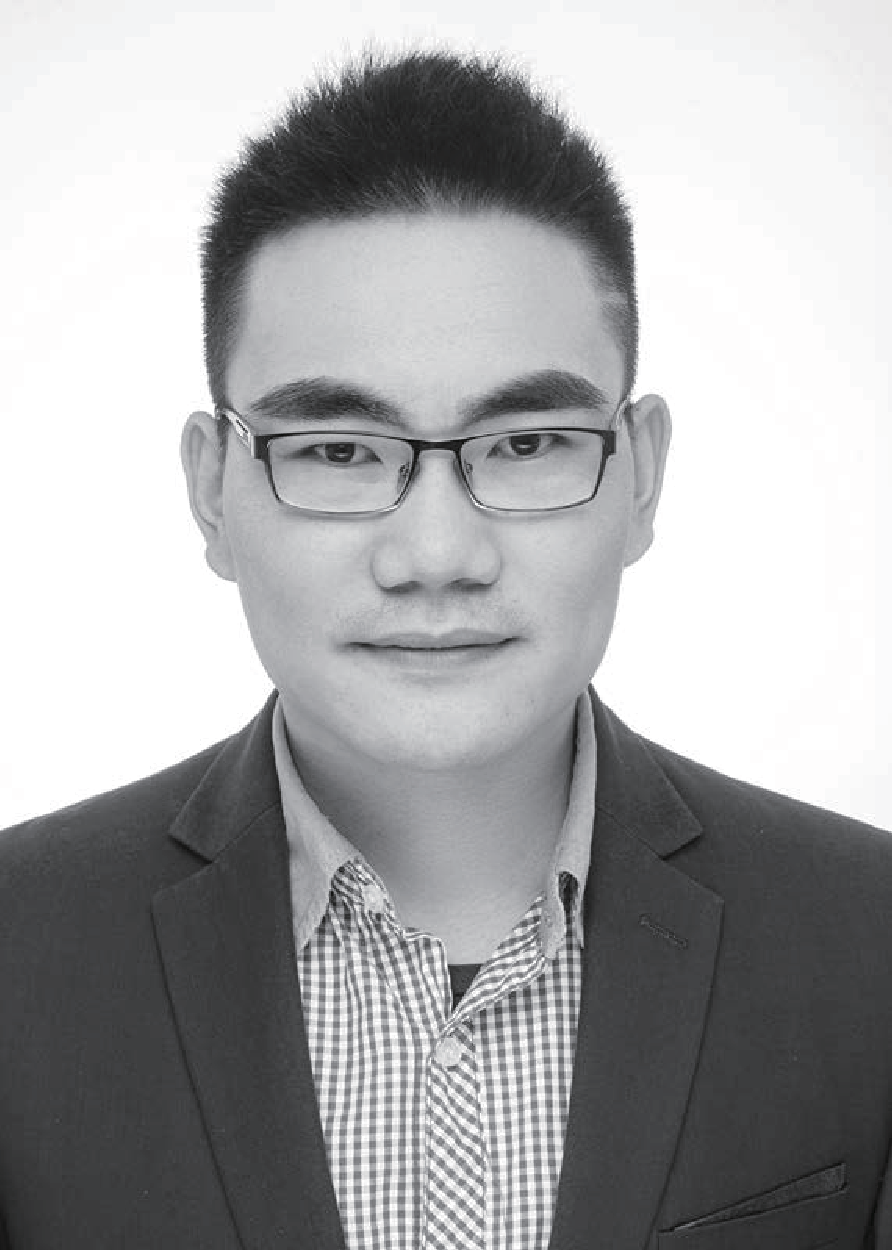}}]{Bowen Zhao}
(Member, IEEE) received the Ph.D. degree in cyberspace security from South China University of Technology, Guangzhou, China, in 2020. 

He was a Research Scientist with the School of Computing and Information Systems, Singapore Management University, from 2020 to 2021. He is currently an Associate Professor with Guangzhou Institute of Technology, Xidian University, Guangzhou. His current research interests include privacy-preserving computation and learning and privacy-preserving crowdsensing.
\end{IEEEbiography}

\begin{IEEEbiography}[{\includegraphics[width=1in,height=1.25in,clip,keepaspectratio]{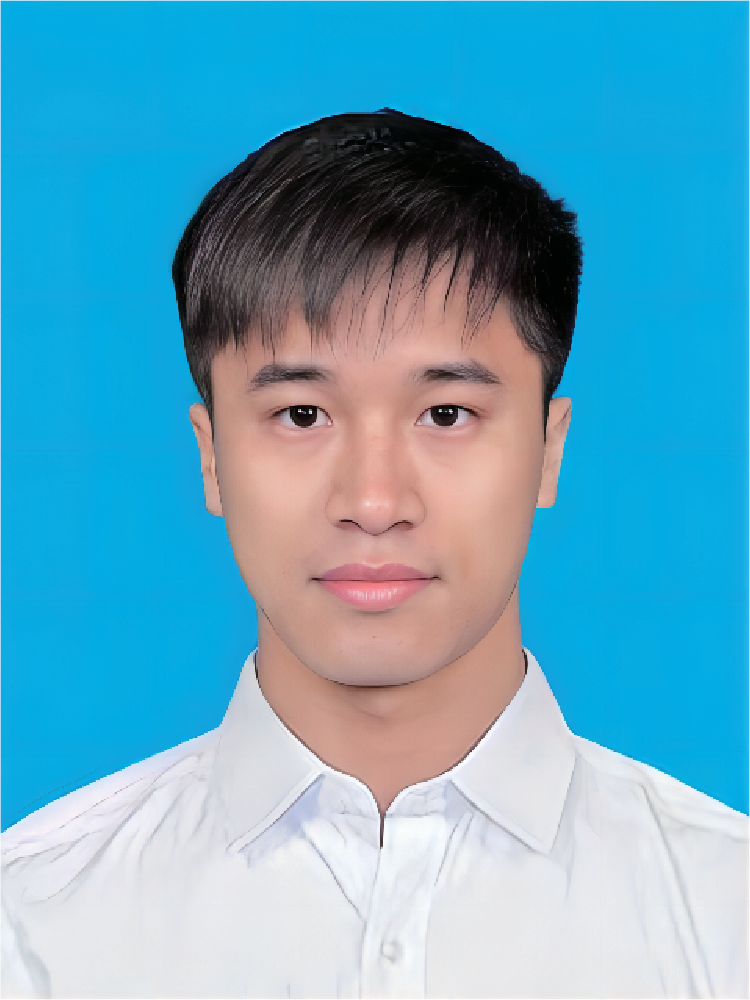}}]{Weiquan Deng}
received the B.S. degree in Information Security from Guangxi University, Nanning, China, in 2023. 

He is currently working toward the M.S. degree in Guangzhou Institute of Technology, Xidian University, Guangzhou. His research interest is privacy-preserving computation.
\end{IEEEbiography}

\begin{IEEEbiography}[{\includegraphics[width=1in,height=1.25in,clip,keepaspectratio]{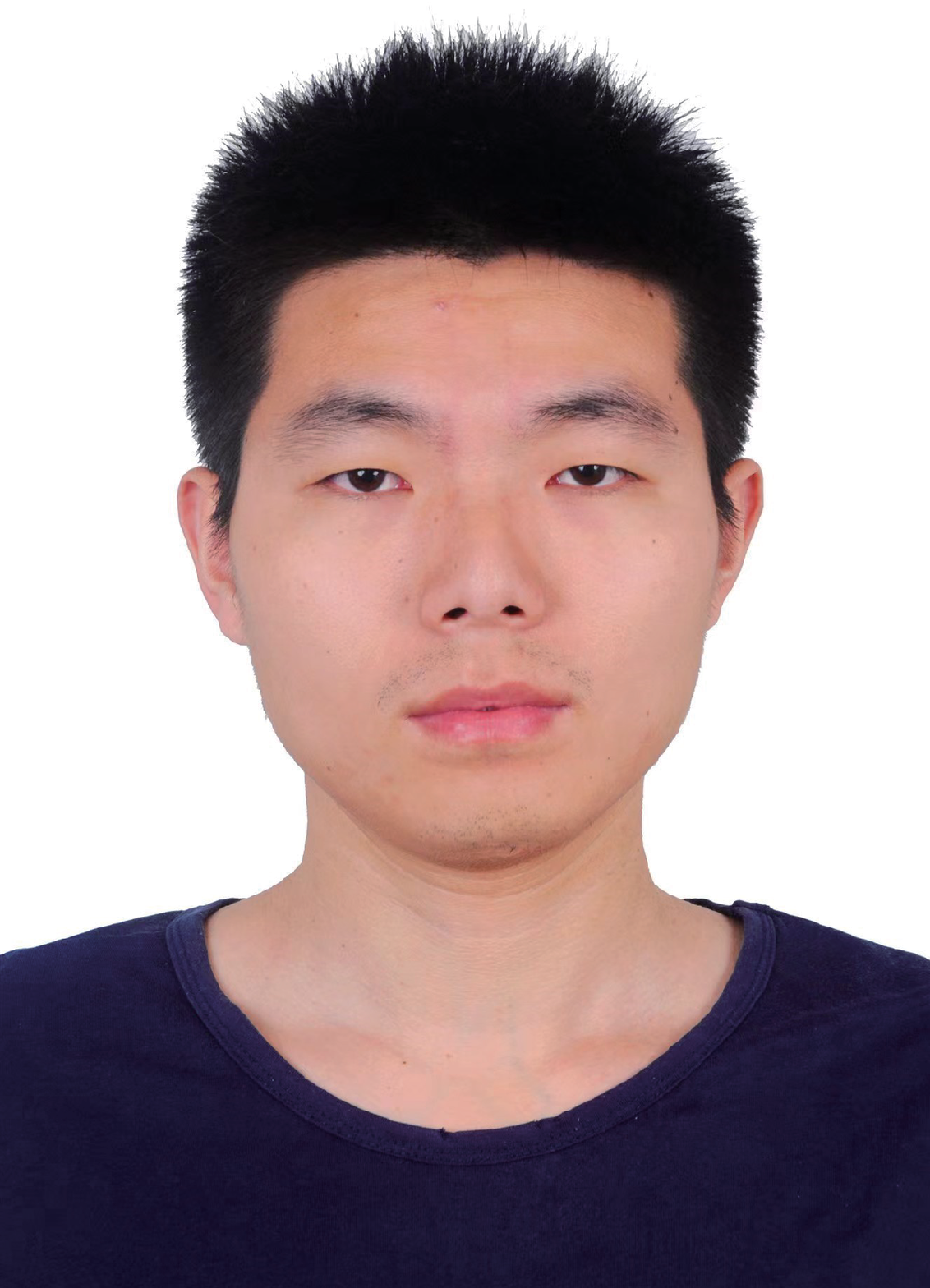}}]{Xiaoguo Li}
received his Ph.D. degree in computer science from Chongqing University, China, in 2019. 

He worked at Hong Kong Baptist University as a Postdoctoral Research Fellow from 2019-2021. He is currently a Research Fellow at Singapre Management University, Singapore. His current research interests include trusted computing, secure computation, and public-key cryptography.
\end{IEEEbiography}

\begin{IEEEbiography}[{\includegraphics[width=1in,height=1.25in,clip,keepaspectratio]{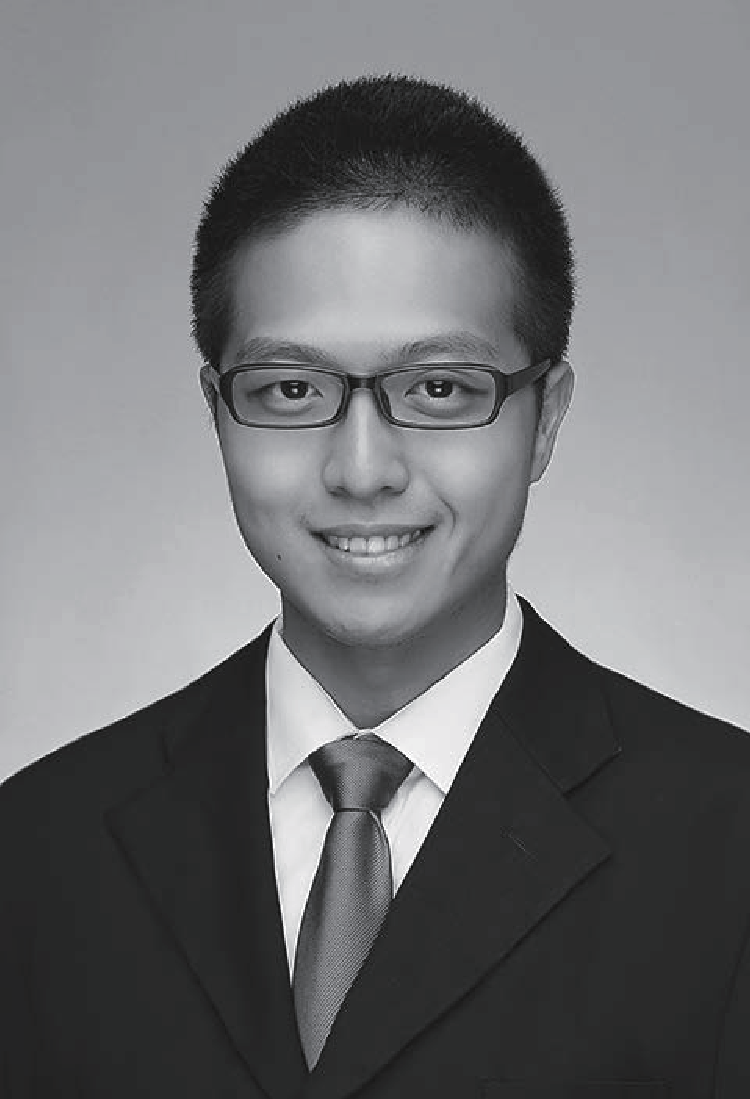}}]{Ximeng Liu}
    (Senior Member, IEEE) received the B.Sc. degree in electronic engineering and the Ph.D. degree in cryptography from Xidian University, Xi’an, China, in 2010 and 2015, respectively. 
    
    He is currently a Full Professor with the College of Computer Science and Data Science, Fuzhou University. He was a Research Fellow with Peng Cheng Laboratory, Shenzhen, China. He has published more than 200 papers on the topics of cloud security and Big Data security including papers in IEEE TOC, IEEE TII, IEEE TDSC, IEEE TSC, IEEE IoT Journal, etc. His research interests include cloud security, applied cryptography and Big Data security. 
    
    Dr. Liu received “Minjiang Scholars” Distinguished Professor Award, “Qishan Scholars” at Fuzhou University and ACM SIGSAC China Rising Star Award (2018).
\end{IEEEbiography}

\begin{IEEEbiography}[{\includegraphics[width=1in,height=1.25in,clip,keepaspectratio]{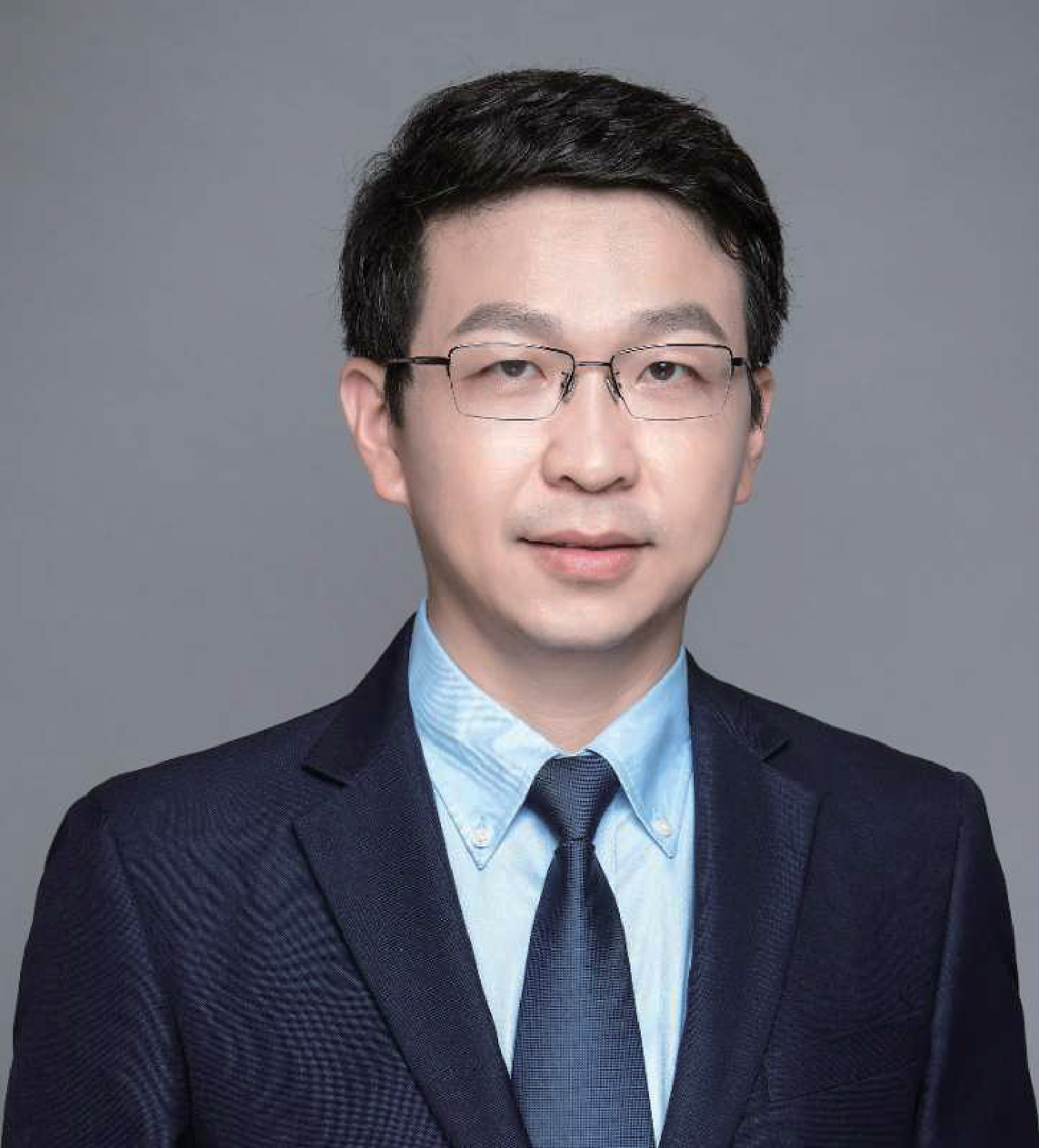}}]{Qingqi Pei}
    (Senior Member, IEEE) received the B.S.,
M.S., and Ph.D. degrees in computer science and
cryptography from Xidian University, Xi’an, China,
in 1998, 2005, and 2008, respectively.

He is currently a Professor and Member of the State
Key Laboratory of Integrated Services Networks, Xidian University. His research interests include digital
content protection and wireless network and security.

Dr. Pei is a Professional Member of the Association
for Computing Machinery and a Senior Member of
the Chinese Institute of Electronics and the China
Computer Federation.
\end{IEEEbiography}

\begin{IEEEbiography}[{\includegraphics[width=1in,height=1.25in,clip,keepaspectratio]{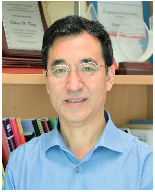}}]{Robert H. Deng}
(Fellow, IEEE) received the Ph.D. degree from the Illinois Institute of Technology, Chicago, IL, USA, in 1985. 

He is AXA Chair Professor of Cybersecurity, Director of the Secure Mobile Centre, School of Computing and Information Systems, Singapore Management University, Singapore. His research interests include applied cryptography, data security and privacy, and network security.
    
\end{IEEEbiography}

\end{document}